\documentclass[12pt]{article}

\usepackage{a4wide}
\usepackage{amsmath,amsfonts,amssymb,
amsthm,color, mathabx}
\usepackage{hyperref}
\usepackage{csquotes}

\usepackage[active]{srcltx}

\usepackage{mathrsfs}
\usepackage{amscd,graphicx}

\usepackage{enumitem}

\let\OLDthebibliography\thebibliography
\renewcommand\thebibliography[1]{
  \OLDthebibliography{#1}
  \setlength{\parskip}{0pt}
  \setlength{\itemsep}{0pt plus 0.3ex}
}

\newtheorem{theorem}{Theorem}[section]

\newtheorem{proposition}[theorem]{Proposition}
\newtheorem{corollary}[theorem]{Corollary}
\newtheorem{lemma}[theorem]{Lemma}

\theoremstyle{definition}
\newtheorem{remark}[theorem]{Remark}

\def\Og{{\mathcal O}}

\def\N{{\mathbb N}}
\def\Z{{\mathbb Z}}
\def\R{{\mathbb R}}

\def\Co{{\mathbb C}}

\def\E{\mathcal E}

\numberwithin{equation}{section}

\newcommand{\sub}[1]{_{\mathrm{#1}}}
\newcommand{\tr}{\mathrm{Tr}}
\newcommand{\di}{\mathrm{d}}
\newcommand{\eu}{\mathrm{e}}
\newcommand{\iu}{\mathrm{i}}

\title{Universality of the Hall conductivity\\ for a weakly interacting magnetic fermionic gas\\ in the Hartree--Fock approximation\par}

\author{H.~D. Cornean\footnote{Department of Mathematical Sciences, Aalborg University, Thomas Manns Vej 23, 9220 Aalborg, Denmark; cornean@math.aau.dk}, E.~L. Giacomelli\footnote{Department of Mathematics ``Federigo Enriques", University of Milan, 
Via Saldini 50, 20133 Milan, Italy; e-mail: emanuela.giacomelli@unimi.it}, D.~Monaco\footnote{Department of Mathematics ``Guido Castelnuovo", Sapienza University of Rome, Piazzale Aldo Moro 5, 
00185 Rome, Italy; email: domenico.monaco@uniroma1.it}, and M.~H. Thorn\footnote{Department of Mathematical Sciences, Aalborg University, Thomas Manns Vej 23, 9220 Aalborg, Denmark; e-mail: mikkelht@math.aau.dk}}

\begin{document}
\maketitle
 
\begin{abstract} 
We consider a two-dimensional gas of interacting fermions in presence of an external constant magnetic field: the system is extended and homogeneous, and thus assumed to be invariant under magnetic translations. Working within the Hartree–Fock approximation, we analyze the system directly in the thermodynamic limit by solving a self-consistent fixed-point equation for the one-particle density matrix. We prove that, provided that the interactions among fermions are sufficiently weak, there exists a unique one-particle density matrix that solves the self-consistency condition. By choosing the Fermi--Dirac distribution as the function in the fixed-point equation, this approach can describe both positive and zero-temperature cases.

At zero temperature and when the chemical potential of the non-interacting system lies in a spectral gap of the free Landau operator, our self-consistent solution is an orthogonal projection (an ``interacting" effective Fermi projection). We prove that its integrated density of states varies linearly with the external magnetic field, provided the interaction is weak enough: the slope of this variation is quantized and independent of the interaction. According to the St\v{r}eda formula, this can be seen as yet another expression of the universality of the quantum Hall effect in weakly interacting systems, at least within the Hartree--Fock approximation.
\end{abstract}

\bigskip

\tableofcontents

\bigskip

\section{Introduction and main results}

In this paper, we adopt Hartree–Fock theory as a first-level approximation for the low-energy properties of two-dimensional weakly interacting fermionic systems in presence of a constant magnetic field, aiming at the derivation of an effective description of the Hall conductivity for interacting fermions. Magnetic pseudo-differential calculus plays an essential role in our work. In order to make the paper (almost) self-contained, we will essentially build-up the whole theory from scratch, adapting the construction to our specific setting with two dimensions, symbols that depend only on momentum, and constant magnetic fields.

\paragraph{Many-body Fermi gases in a magnetic field.}

Let us start by considering a two-dimensional homogeneous gas of $N$ weakly interacting fermions confined to the box $\Lambda = [-L_1/2, L_1/2]\times [-L_2/2, L_2/2] \subset \mathbb{R}^2$ endowed with magnetic-periodic boundary conditions (we will later pass to the infinite-volume, thermodynamic limit). Spin degrees of freedom are neglected, and all fermions are assumed to have the same spin. The particles interact via a weak two-body potential, whose intensity is specified through a coupling constant $\lambda \ll 1$, and are exposed to a uniform external magnetic field of strength $b\in\mathbb{R}$, perpendicular to $\Lambda$. The corresponding many-body Hamiltonian is given by
\begin{equation}\label{eq: Hnblambda}
    H_{N,b}^{\Lambda} \coloneq \sum_{i=1}^{N} \left(H_b^{\Lambda}\right)_{x_i} + 2 \, \lambda \, \sum_{1 \le i < j \le N} V(x_i; x_j),
\end{equation}
where the one-particle magnetic Hamiltonian acts as
\begin{equation}\label{eq: def Hblambda}
    (H_{b}^\Lambda)_{x_i}= \big( - \iu \nabla_{x_i} - b \, A(x_i)\big)^2, \quad A(x_i) = A(x_i^{(1)},x_i^{(2)}) \coloneq \frac{1}{2} (-x_i^{(2)}, x_i^{(1)}).
\end{equation}
Due to the Pauli principle, the Hamiltonian $H_{N,b}^{\Lambda}$ acts on the antisymmetric subspace $L^2\sub{a}(\Lambda^{N}) \subset L^2(\Lambda)^{\otimes N}$. It is well known that the free magnetic Hamiltonian $H_b^\Lambda$ commutes with magnetic translations, i.e., transformations $\tau_{b,(L_1,0)}, \tau_{b,(0,L_2)}\colon L^2(\Lambda) \rightarrow L^2(\Lambda)$ acting on $\psi\in L^2(\Lambda)$ as
\[ (\tau_{b,(L_1,0)}\psi)(x_1,x_2) \coloneq \eu^{-\iu \, b \, L_1 \, x_2} \, \psi(x_1-L_1, x_2), \quad (\tau_{b,(0,L_2)}\psi)(x_1,x_2) \coloneq \eu^{\iu \, b \, L_2 \, x_1} \, \psi(x_1, x_2-L_2)\,. \]
If the side lengths $L_1, L_2>0$ satisfy the flux quantization condition
\begin{equation} \label{eqn:FluXQuantization} 
b \, L_1 \, L_2 = 2 \pi \, M\,, \quad M \in \Z\,,
\end{equation}
then the magnetic translations  $\tau_{b, (L_1, 0)}$ and $\tau_{b, (L_2, 0)}$ commute among themselves. Throughout this introduction, we assume \eqref{eqn:FluXQuantization} holds. As mentioned above, magnetic-periodic boundary conditions are imposed at the boundary of $\Lambda$, namely
\[ \psi(L_1/2, x_2) = (\tau_{b,(L_1,0)}\psi)(-L_1/2,x_2)\,, \quad \psi(x_1,L_2/2) = (\tau_{b,(0,L_2)}\psi)(x_1,-L_2/2)\,.\]
Under these conditions, the operator $H_b^\Lambda$ is self-adjoint operator on $L^2(\Lambda)$, and its spectrum consists of the {Landau levels}
\[ E_n \coloneq |b| \, (2\,n + 1)\,, \quad n \in \N\,,\]
each having degeneracy $|M|$, where $M$ is as in~\eqref{eqn:FluXQuantization} \cite{FerreroMonaco2024}.

The two-body interaction potential is assumed to be compatible with the magnetic-periodic boundary conditions and is taken of the form $V(x,y) = v(x-y)$, where $v\colon \mathbb{R}^2 \to \mathbb{R}$ is an even function satisfying suitable regularity and decay assumptions specified later.

\paragraph{Hartree--Fock theory.}

It is well known that the low-energy properties of a fermionic gas can be described, as a first approximation, within \emph{Hartree--Fock theory}. Suppose the state of the gas is described by a density matrix $\omega$ on $L^2\sub{a}(\Lambda^N)$: the energy of the system is then computed as 
\[ E_{N,b}^{\Lambda}[\omega] = \tr_{L^2\sub{a}(\Lambda^N)}\left( H^\Lambda_{N,b} \, \omega \right)\,. \]
The Hartree--Fock approximation amounts to restricting the many-body state space to \emph{quasi-free states}. For this class of states, all correlation functions (including the energy) are determined by the \emph{one-particle reduced density matrix} $\gamma$, which is a self-adjoint operator on $L^2(\Lambda)$ satisfying the Pauli constraint
\begin{equation}\label{eq: constraintis gamma}
0 \le \gamma \le 1, \qquad \operatorname{Tr}_{L^2(\Lambda)}\gamma = N = \rho |\Lambda|
\end{equation}
where $\rho$ denotes the particle density. Among these states are Slater determinants, i.e., antisymmetric products of mutually orthogonal single-particle states $\{ f_1, \ldots, f_N\} \subset L^2(\Lambda)$ of the form
\[
\omega_{\mathrm{Slater}} \coloneq \left| \psi_{\mathrm{Slater}} \right\rangle \left\langle \psi_{\mathrm{Slater}} \right|\,, \quad \psi_{\mathrm{Slater}} = f_1 \wedge \cdots \wedge f_N,
\]
for which the one-particle reduced density matrix $\gamma$ is a rank-$N$ projection and satisfies $\gamma^2 = \gamma$. Quasi-free states that are not Slater determinants are called mixed quasi-free states.

Upon restriction to such states, the energy of the system is represented by the \emph{Hartree--Fock functional}, which depends on the reduced one-particle density matrix $\gamma$ through its integral kernel $\gamma(x;y)$ as
\[ \E\sub{HF}^{\Lambda,b}[\gamma] \coloneq \tr_{L^2(\Lambda)}(H_b^{\Lambda} \, \gamma) + \lambda \int_{\Lambda} \di x \int_{\Lambda} \di y \, v(x-y) \big[ \gamma(x,x) \, \gamma(y,y) - \left| \gamma(x,y)\right|^2 \big]\,. \]
For non-interacting fermions ($\lambda = 0$), the Hartree--Fock approximation is exact: $E_{N,b}^{\Lambda}[\omega] =  \E\sub{HF}^{\Lambda,b}[\gamma]$ for $\gamma = \omega^{(1)}$ the one-particle reduced density matrix obtained from $\omega$ by ``tracing out'' the degrees of freedom associated to $N-1$ particles. As an example, at inverse temperature $\beta>0$ and chemical potential $\mu\in\mathbb{R}$, the reduced one-particle density matrix $\gamma_{\beta, \mu}$ associated to a thermal Gibbs state $\omega_{\beta, \mu}$ has the form 
\[
    \gamma_{\beta, \mu} = f_{\mathrm{FD}}(H^{\Lambda}_b) = e^{-\beta(H^\Lambda_b -\mu)} \left(\mathrm{Id} + e^{-\beta(H^\Lambda_b -\mu)} \right)^{-1} = \left(\mathrm{Id} + e^{\beta(H^\Lambda_b -\mu)} \right)^{-1},
\]
where $f_{\mathrm{FD}}(\cdot)$ denotes the \emph{Fermi--Dirac distribution} \cite[Prop. 5.2.23]{BR97}
\[
    f_{\mathrm{FD}}(t)\coloneq \frac{1}{1+e^{\beta(t-\mu)}}, \quad t\in\mathbb{R}.
\]
An important problem in the analysis of many-body quantum systems is to go beyond the non-interacting regime and analyze the range of validity of Hartree--Fock theory accounting for correlation effects, particularly in the study of ground-state energies (see, e.g., \cite{Bac,GS94}). Significant progress has been made in recent years in the understanding of the correlation energy of Fermi gases, both in the low- and high-density regimes: see \cite{GHNS25, CHN24, BPSS22, CZ25} and references therein. A corresponding derivation of time-dependent Hartree--Fock equations from the many-body Schr\"{o}dinger dynamics \cite{Lubich} has been recently established also in presence of magnetic fields, for appropriate time scales and strength of the two-body interactions \cite{FerreroMonaco2024, BBMN25}. The study of the Hartree--Fock approximation of electronic structure is also pivotal for density functional theory (DFT) \cite{CF24}, which has been recently investigated also in low-dimensional magnetic materials \cite{GL19, GLM23}. DFT models for interacting fermions in magnetic fields have been studied also in other regimes: addressing 3-dimensional systems, modeling the coupling of spin to magnetic fields via the Pauli Hamiltonian, or confining the system in a potential trap. Within these settings, the ground-state energy asymptotics has been related to magnetic Thomas–Fermi limits: for a selection of results along these lines, the reader is referred e.g.\ to \cite{Y91, LSY95, LMT19, FM20} and references therein.

Coming back to the Hartree--Fock energy functional for general quasi-free states, we rewrite it as follows:
\begin{equation} \label{eqn:EHF-W}
\begin{aligned}
\E\sub{HF}^{\Lambda,b}[\gamma] & = \tr_{L^2(\Lambda)}(H_b^{\Lambda} \, \gamma) + \lambda \, \int_{\Lambda} \di x \left[\int \di x' \, \delta(x-x') \, \int_{\Lambda} \di y \, v(x-y) \, \gamma(y,y)] \right] \, \gamma(x',x) \\
& \quad - \lambda \, \int_{\Lambda} \di x \left[\int_{\Lambda} \di x' \, v(x-x') \, \gamma(x,x') \right] \, \overline{\gamma(x,x')} \\
& \eqcolon \tr_{L^2(\Lambda)}(H_b^{\Lambda} \, \gamma) + \lambda \, \int_{\Lambda} \di x \, W^{\Lambda}[\gamma](x,x') \, \gamma(x',x) \\
& = \tr_{L^2(\Lambda)} \big( \left(H_b^{\Lambda} + \lambda \, W^{\Lambda}[\gamma] \right) \, \gamma \big)\,,
\end{aligned}
\end{equation}
where we have set
\begin{equation} \label{eqn:W}
W^{\Lambda}[\gamma](x,x') \coloneq \delta(x-x') \, \int_{\Lambda} v(x-y) \, \gamma(y,y) \, \di y- v(x-x') \, \gamma(x,x')\,. 
\end{equation}

For each density matrix $\gamma$ on $L^2(\Lambda)$, $W^{\Lambda}[\gamma]$ defines a self-adjoint operator whose integral kernel is specified by~\eqref{eqn:W}. Provided this operator is bounded, setting 
\[ H_{b,\,\lambda}^{\Lambda}[\gamma] \coloneq H_b^{\Lambda} + \lambda \,W^{\Lambda}[\gamma] \]
also defines a self-adjoint operator, which we can consider as an effective one-particle Hamiltonian, depending on the density matrix $\gamma$.

\paragraph{Self-consistent Hartree--Fock density matrices.}

As seen before, in the non-inter\-acting regime where the Hartree--Fock approximation is exact, the one-particle reduced density matrix is typically a function (e.g.\ the Fermi--Dirac distribution) of the one-body Hamiltonian. In electronic structure theory, it is often postulated that the same holds for weakly interacting fermionic systems: such effective non-linear Hamiltonians, and specifically the effective $1$-body potential $W^{\Lambda}$, are determined within Hartree--Fock theory through a self-consistency condition \cite{FW71, CF24}, assuming that low-energy states of the system are described by a one-particle reduced density matrix  $\gamma_\star$ of the form
\begin{equation} \label{eqn:fixed-point-gammastar} 
\gamma_\star = f\left({H}_{b,\lambda}^{\Lambda}[\gamma_\star]\right)\,.
\end{equation}
The above should be interpreted as a fixed-point problem for the operator $\gamma_\star$.

In this paper, we study the existence of self-adjoint operators $\gamma_\star$ satisfying \eqref{eq: constraintis gamma} and solving an equation of the form \eqref{eqn:fixed-point-gammastar}, for an appropriate class of functions $f$. We formulate the question directly for extended systems, that is, in the thermodynamic, infinite-volume limit. For (magnetic-)translation-invariant density matrices $\gamma$, the above Hartree--Fock energy $\E\sub{HF}^{\Lambda}[\gamma]$ is extensive, i.e., it grows proportionally to the area $|\Lambda| = L_1 \, L_2$ of the confining box. It is therefore to be expected that the limit
\[ \E\sub{HF}[\gamma] \coloneq \lim_{L_1, L_2 \to \infty} \frac{1}{|\Lambda|} \,\E\sub{HF}^{\Lambda}[\gamma] \]
exists and is finite: one should take the limit along sequences for which the quantization condition $b L_1 L_2 \in 2 \pi \Z$ from \eqref{eqn:FluXQuantization} is satisfied, and again at fixed density 
\[ \rho = \rho_\gamma \coloneq \lim_{L_1, L_2 \to \infty} \frac{1}{|\Lambda|} \, \tr_{L^2(\Lambda)}(\gamma) = \gamma(0,0) \,.  \]
In this limit, we expect the trace in \eqref{eqn:EHF-W} to become a trace per unit area, $\underline{\tr}$, and thus
\[  \E\sub{HF}[\gamma] = \underline{\tr} \big( \left(H_b + \lambda \, W[\gamma] \right) \, \gamma \big)\,, \]
where $H_b$ is the operator corresponding to \eqref{eq: def Hblambda} and defined in $\mathbb{R}^2$, i.e. 
\begin{equation}\label{eq: def Hb}
    H_{b}\coloneq \big( - \iu \nabla - b \, A(x)\big)^2, \quad A(x) = A(x_1,x_2) = \frac{1}{2} (-x_2, x_1),
\end{equation}
and where $W[\gamma]$ is the operator on $L^2(\R^2)$ specified by the kernel
\[
W[\gamma](x,x') \coloneq \delta(x-x') \, v(x) \, \rho_\gamma - v(x-x') \, \gamma(x,x')\,, 
\]
with $x,x^\prime\in\mathbb{R}^2$.
Notice that the above operator, to be interpreted as an effective potential, is in general {not} a multiplication operator, and actually $H^{\Lambda}_{b,\,\lambda}$ will in general not even be a differential operator---rather, a pseudo-differential one.

While we do not prove the above statements about the thermodynamic limit, and we leave this analysis to future work, we study next the problem of finding the self-consistent translation-invariant density matrix~$\gamma_\star$, defined by the analogue of \eqref{eqn:fixed-point-gammastar} directly in the infinite-volume regime.

\subsection{Main results}\label{subsec:main}

We are interested in formulating the self-consistency equations~\eqref{eqn:W} and \eqref{eqn:fixed-point-gammastar} directly in the infinite-volume limit, namely
\begin{equation}\label{0}
\begin{aligned}
W(x,x')& = \delta (x-x')\int_{\R^2} v(x-y)\, f (H_b+ \lambda W)(y,y)\di y\\ 
&\qquad - v(x-x')\, f (H_b+ \lambda W)(x,x')\,,\quad x,x'\in \R^2,
\end{aligned}
\end{equation}
where the self-consistent operator $W$ specified by the above kernel should be bounded and self-adjoint. 

In order to do so, we use some ideas from \emph{gauge-covariant magnetic perturbation theory} \cite{CorneanNenciu1998,Nenciu2002,CorneanMoscolari2025}. In particular, we want to exploit the fact that the density matrix $f(H_b+ \lambda W)$ commutes with all possible magnetic translations, which restricts the form of its integral kernel: stated informally, the latter is the pointwise product of a \emph{magnetic Peierls phase} $\eu^{\iu \, b \, \phi(\cdot,\cdot)}$ and a translation-invariant part depending only on the difference of the two arguments of the kernel. To state this more precisely, define first the magnetic phase argument 
\begin{equation}\label{pierls}
\phi(x,x')\coloneq(x_2x_1'-x_1x_2')/2 \eqcolon \frac{1}{2}\, x'\wedge x,\quad  x,x' \in\mathbb{R}^2\,.
\end{equation}
Note that $\phi(x,x')=-\phi(x',x)$. The magnetic translation operators are defined for $y \in \R^2$ as
\begin{equation} \label{eqn:MagTransl}
(\tau_{b,y}\,\psi)(x) \coloneq \eu^{\iu \, b \, \phi(x,y)} \, \psi(x-y)\,, \quad \psi \in L^2(\R^2)\,,
\end{equation}
We then get the following result with an accompanying corollary: their proofs are postponed to Section~\ref{sec:kernels}.

\begin{proposition} \label{prop:MagCommuteKernel} 
Let $T$ be a bounded integral operator on $L^2(\R^2)$ with an integral kernel $T(x,x')$ which is locally integrable separately in $x$ and $x'$. Then, $T$ commutes with all magnetic translations \eqref{eqn:MagTransl} if and only if 
\begin{equation*}
T(x,x') = \eu^{\iu \, b \, \phi(x,x')} \, F_T(x-x')\,, \quad \text{where} \quad F_T(x) \coloneq T(x,0)\,.
\end{equation*}
\end{proposition}

\begin{corollary} \label{cor:MagCommuteKernel} 
If $T$ is as in Proposition~\ref{prop:MagCommuteKernel}, then
\begin{enumerate}[label={\rm(\roman*)}, ref={\rm(\roman*)}]
    \item provided $F_T$ is continuous at zero, $T(x,x) = F_T(0)$ is independent of $x \in \R^2$;
    \item $T = T^*$ if and only if $\overline{F_T(-x)} = F_T(x)$ for $x \in \R^2$.
\end{enumerate}
\end{corollary} 

With this observation in mind and assuming as an \textsl{ansatz} that the operators $W$ and $f(H_b+\lambda W)$ commute with all magnetic translations, we can restate \eqref{0} as follows: find a function $F$ such that 
\begin{equation}\label{1}
\begin{aligned}
    &\qquad \qquad \qquad \qquad F(x)=f(H_b+\lambda W_{F})(x,0),\\
    W_{F}&\coloneq \left( \int_{\R^2} v(y) \, \di y \right) \, F(0) + Z_{F}\, ,\quad Z_{F}(x,x') \coloneq - \eu^{\iu \, b \, \phi(x,x')} \, v(x-x') \, F(x-x') .
    \end{aligned}
\end{equation}
In Section~\ref{sec:kernels} and Section~\ref{sec:fixed} we show that the above objects are well-defined and the equation~\eqref{1} is well-posed, when the function $v$ from the two-body interaction $V(x;y)=v(x-y)$ is any even real function such that%
\footnote{Recall the Japanese bracket notation $\langle x \rangle \coloneq (1+|x|^2)^{1/2}$.} %
$\langle x\rangle ^n \, v(x)\in L^1(\R^2)$ for all $n\geq 0$. For concreteness, $v$ can be chosen as the screened Coulomb potential in 2D, i.e.\ $v(x)=e^{-\alpha|x|}\,\ln(|x|)$ with $\alpha>0$. Correspondingly, we allow functions $F$ as elements of the set
\begin{equation}\label{cdh1} BC_H(\R^2)\coloneq\{G\in BC(\R^2)\mid \overline{G(-x)} = G(x)\,,\:\forall x\in\R^2\}\,,
\end{equation}
where $BC(\R^2) \coloneq C^0(\R^2) \cap L^\infty(\R^2)$ is the set of complex-valued bounded and continuous functions on $\R^2$. 

Concerning the function $f$ on the right hand side of \eqref{1}, and having in mind the Fermi--Dirac distribution, we will ask it to be smooth and with fast decay at infinity.  This corresponds to considering an $f$ which is real-valued and that for any $n,k\geq 0$ satisfies
\begin{equation}\label{cdh2}
\sup_{x\in[0,\infty)}\langle x\rangle^n|f^{(k)}(x)|<\infty.
\end{equation}

We may now formulate the first main result of our paper.
\begin{theorem}\label{thm1}
    The following properties hold true: 

\begin{enumerate}[label={\rm(\roman*)}, ref={\rm(\roman*)}]
 \item  \label{thm1-i} For any $F\in BC_H(\R^2)$, the operator $W_F$ defined in \eqref{1} is self-adjoint and bounded, so that in particular
 \begin{equation} \label{eqn:HblF}
  H_{b,\lambda,F} \coloneq H_b + \lambda \, W_F, \quad \lambda \in \R\,, 
 \end{equation}
 is a lower-bounded self-adjoint operator with the same domain as $H_b$. 
 \item \label{thm1-ii} Fix $f$ obeying \eqref{cdh2}. The operator $f(H_{b,\lambda,F})$, defined by functional calculus, has a jointly continuous integral kernel $f(H_{b,\lambda,F})(x,x')$. Moreover, for every $F\in BC_H(\R^2)$ we have 
 \[ \Phi_{b,\lambda}^{f}(F)(\cdot ) \equiv \Phi_\lambda(F)(\cdot )\coloneq f(H_{b,\lambda,F})(\cdot,0) \in BC_H(\R^2)\,. \]
 \item \label{thm1-iii} There exists a $\lambda_0>0$ depending on $b$ and $f$ such that the fixed-point equation
 \begin{equation} \label{eqn:FixedPoint} 
 \Phi_\lambda(F) = F\,, \quad F \in BC_H(\R^2)\,,
 \end{equation}
 (cf.\ \eqref{1}) has a solution for all $|\lambda|\leq \lambda_0$. 
\end{enumerate}
 \end{theorem}

\begin{remark} \label{rmk:mainThm1}
 The proof of Theorem~\ref{thm1} is covered by the results in Section~\ref{sec:kernels} and \ref{sec:fixed}. The main strategy is to work with the ``magnetic" symbols of the operators involved (see Section~\ref{sec:pseudodiffs}) instead of their integral kernels. 
 
 An important fact is that the resolvent of $H_b$, as well as the resolvent of some ``small" perturbations of $H_b$, have symbols with a certain decay, see Lemma~\ref{pseudo:resolvent2}, which translates into integral kernels having appropriate regularity. Combining this with magnetic pseudo-differential calculus and the properties of $f$, we can conclude that $f(H_{b,\lambda,F})$ has a smoothing symbol, as seen in Proposition~\ref{prop:fermi_dirac_kernel}. 
 
 Furthermore, when $|\lambda|$ is sufficiently small, then the difference in symbols for $f(H_{b,\lambda,F})$ and $f(H_{b,\lambda,G})$ behaves asymptotically as $\Og(|\lambda|\,\Vert F-G\Vert_{L^\infty(\R^2)})$ for $F,G$ in a bounded set (see \eqref{hdc12}). This sets us up to use the contraction mapping theorem, also known as Banach's fixed point theorem, in Corollary~\ref{thm:fix_point_final}, which also gives uniqueness of the solution to \eqref{eqn:FixedPoint} in $L^\infty$-balls in $BC_H(\R^2)$.
\end{remark}

The second main result is related to the so-called \emph{zero-temperature limit}. It concerns the self-consistent Fermi projection $P$, obtained from the above Theorem as the solution of the fixed-point equation corresponding to $f$'s which are smoothed characteristic functions of the interval $(-\infty, \mu)$ with $\mu$ in a spectral gap of $H_b$. The result states that, if the interaction is weak compared to the external magnetic field so that the gap which contains the chemical potential $\mu$ does not close, then $P$ has an \emph{integrated density of states} (which in view of the results of Proposition~\ref{prop:MagCommuteKernel} coincides with the constant value along the diagonal of its integral kernel) which grows linearly in $b$ with a quantized slope. 

\begin{theorem}\label{thm2}
Let us fix $b>0$ and let us choose any real-valued, compactly supported function $f_N$ which equals $1$ near $N<\infty$ Landau levels $\{(2n+1)b \mid n\geq 0\}$ and equals zero near all others. Let $\lambda_0>0$ be as in Theorem~\ref{thm1} \ref{thm1-iii} and denote by $F_{b,\lambda}$ the fixed point of $\Phi_\lambda$ for $|\lambda|\leq\lambda_0$.
Then:
\begin{enumerate}[label={\rm(\roman*)}, ref={\rm(\roman*)}]
    \item \label{thm2-i} There exists $0<\lambda_1\leq \lambda_0$ such that for all $|\lambda|\leq\lambda_1$ the operator $P_{b,\lambda,N}\coloneq f_N(H_{b,\lambda,F_{b,\lambda}})$ is an orthogonal projection. 
    \item \label{thm2-ii} There exists some $0<\lambda_2\leq \lambda_1$  such that for all $|\lambda|\leq \lambda_2$ we have that the integrated density of states of $P_{b,\lambda,N}$ is independent of $\lambda$ and equals   
    \begin{equation} \label{Streda}
    \mathcal{I}(P_{b,\lambda,N}) = P_{b,\lambda,N}(0,0)=\frac{N\, b}{2\pi}\, .
    \end{equation}
\end{enumerate}
\end{theorem}

We prove Theorem~\ref{thm2} in Section~\ref{sec:zerotemp}. The proof uses the strategy from \cite{CorneanMonacoMoscolari2021,CorneanMoscolari2025}. The main idea is as follows: if $\lambda=0$, the integrated density of states of $P_{b,0,N}\coloneq f_N(H_b)$ is well known and equals $\frac{N\, b}{2\pi}$ \cite{Landau1930}. One also proves that
\[P_{b,\lambda,N}-P_{b,0,N}=f_N(H_{b,\lambda,F_{b,\lambda}})-f_N(H_{b,\lambda,0})\sim |\lambda|\, ,\]
hence, when $|\lambda|$ is small, there exists an intertwining unitary between $P_{b,\lambda,N}$ and $P_{b,0,N}$ having some extra special properties which ensure that the two projections must have the same integrated density of states. This is done in Proposition~\ref{prop:traces}.

From \eqref{Streda}, one trivially obtains that for all sufficiently small $\lambda$'s
\[ 2 \pi \, \frac{\partial \mathcal{I}(P_{b,\lambda,N})}{\partial b} = 2 \pi \, \frac{\partial \mathcal{I}(P_{b,0,N})}{\partial b} = N\,. \]
It was realized by St\v{r}eda \cite{Streda1982} that the derivative of the integrated density of states with respect to the magnetic field is proportional to the \emph{Hall conductivity} of the 2D magnetic fermion gas, in appropriate physical units (equal to $1/2\pi$ in our notation); see \cite{CorneanNenciuPedersen2006, CorneanNenciu2009} for proofs of this statement in the non-interacting approximation. Our second main result, Theorem~\ref{thm2}, thus shows that the self-consistent interacting Hartree--Fock projection carries the same Hall conductivity of the non-interacting one, provided the interaction strength $\lambda$ is small enough; moreover, its value is \emph{quantized} to an integer, which in the above Theorem is the number of occupied Landau levels, but that in general coincides with the \emph{Chern number} of the Fermi projection \cite{CorneanMonacoMoscolari2021}. This proves the Hall conductivity to be a \emph{universal} topological transport coefficient in weakly interacting fermionic systems, at least within the Hartree--Fock approximation. This result is in line with several similar statements in the recent literature, obtained by different methods and without resorting to Hartree--Fock theory, in weakly interacting fermionic systems on a lattice, showing quantization and proving universality of the Hall conductance \cite{HastingsMichalakis2015, GiulianiMastropietroPorta2017, BachmannEtAl2018} or of the Hall conductivity \cite{WesleEtAl2025}.

\section{Preliminaries: 2D magnetic pseudo-differential operators with constant magnetic field}\label{sec:pseudodiffs}

As stated in Remark~\ref{rmk:mainThm1}, our proofs rely on magnetic pseudo-differential calculus, which for the sake of completeness will be presented in this Section. In our setting, the magnetic field is constant and all the symbols only depend on the momentum, and not on the position variable. This leads to a significant simplification of the theory: many important results like the Calder{\' o}n-Vaillancourt theorem or Beals' commutator criterion admit simpler proofs and formulations in this framework. Most of the results in this section can be found, in a more general form, in \cite{Mantoiu2004,Iftimie2007,Iftimie2019,CorneanHelfferPurice2018,CorneanHelfferPurice2024}. Neverthless, we decided to reformulate (almost) the whole theory in our simplified setting,  in order to make our paper (almost) self-contained. If the reader is familiar with \cite{Mantoiu2004,CorneanHelfferPurice2018,CorneanHelfferPurice2024}, then they may skip forward to Subsection~\ref{subsec:resolvents}.

We denote by $S^{m}_0(\R^2)$, $m\in\Z$, the set of functions $p\in C^\infty(\R^2)$ such that for any $\alpha=(\alpha_1,\alpha_2)\in \Z_+^2$ we have a constant $C_\alpha$ such that
\begin{equation}\label{pseudo1}
    |\partial^\alpha p(\xi)|\leq C_\alpha\, \langle \xi\rangle^{m},\quad \forall \xi\in\R^2.
\end{equation}
We define a class of magnetic pseudo-differential operators by 
\begin{equation}\label{pseudo2}
\langle f_1, {\rm Op}^A(p) f_2\rangle \coloneq(2\pi)^{-2} \int_{\R^2}\, \di\xi\, \int_{\R^4}\di x\, \di y\, \overline{f_1(x)}\, f_2(y)\, \eu^{\iu b\phi(x,y)} \eu^{\iu\xi\cdot (x-y)} p(\xi),
\end{equation}
where $f_1,f_2\in\mathscr{S}(\R^2)$ are Schwartz functions and
\begin{equation*}
 \phi(x,y)=(x_2y_1-x_1y_2)/2 =x\cdot y^\perp,
\end{equation*}
is as in \eqref{pierls}, with $y^\perp=(-y_2/2,y_1/2)$. We call $p$ the symbol of the operator ${\rm Op}^A(p)$.

\subsection{A tight magnetic Gabor frame}

The following will be a technical tool used throughout this section. See also Lemma~2.1 \cite{CorneanHelfferPurice2018} and Proposition~2.2 \cite{CorneanHelfferPurice2024}.

\begin{lemma}[Gabor frame]\label{pseudo:frame}
    Let $\gamma, \gamma^*$ be in $\mathbb{Z}^2$ and let $\phi(\cdot, \cdot)$ be as in \eqref{pierls}.
    Consider the set of functions  
    \[\Psi_{\gamma,\gamma^*}(y)\coloneq\eu^{\iu b\phi(y,\gamma)}g(y-\gamma)(2\pi)^{-1} \eu^{\iu \gamma^* \cdot (y-\gamma)}, \qquad y\in\mathbb{R}^2
    \] 
    with $g\in C_0^\infty(\R^2)$ such that
    \begin{equation}\label{eq:prop g}
    \mathrm{supp}(g)\subset (-1,1)^2\subset\mathbb{R}^2, \qquad \sum_{\gamma\in \Z^2} g^2(x-\gamma)=1.
    \end{equation}
    Then the $\Psi_{\gamma,\gamma^*}$'s constitute a tight Parseval frame in $L^2(\R^2)$ and for $f\in\mathscr{S}(\R^2)$ we have
    \begin{equation*}
        f(x)=\sum_{\gamma,\gamma^*\in\Z^2} \Psi_{\gamma,\gamma^*}(x)\, \langle \Psi_{\gamma,\gamma^*},f\rangle,
    \end{equation*}
    with convergence in $\mathscr{S}(\R^2)$.
\end{lemma}

\begin{remark}[Gabor frame in $L^2(\mathbb{R}^2)$]\label{rem: Gabor} Note that being a tight Parseval frame implies that the decomposition
\begin{equation*}
    f=\sum_{\gamma,\gamma^*\in\Z^2} \Psi_{\gamma,\gamma^*}\, \langle \Psi_{\gamma,\gamma^*},f\rangle,
\end{equation*}
holds for every $f\in L^2(\R^2)$, see Corollary~5.1.7 \cite{Christensen2016} or compare with the proof of Proposition~2.2 in \cite{CorneanHelfferPurice2024}.
\end{remark}
\begin{proof}[Proof of Lemma~\ref{pseudo:frame}]
    For $f_1,f_2\in L^2(\R^2)$ we have, using the properties of $g$ \eqref{eq:prop g},
    \[\langle f_1,f_2\rangle=\sum_{\gamma\in\Z^2}\langle g(\cdot-\gamma)f_1,g(\cdot-\gamma)f_2\rangle=\sum_{\gamma\in\Z^2}\langle \eu^{\iu b\phi(\cdot,\gamma)}g(\cdot-\gamma)f_1,\eu^{\iu b\phi(\cdot,\gamma)}g(\cdot-\gamma)f_2\rangle,\]
    and then, expanding each term in a Fourier series,
    \[\langle f_1,f_2\rangle=\sum_{\gamma\in\Z^2}\langle f_1,\Psi_{\gamma,\gamma^*}\rangle\langle\Psi_{\gamma,\gamma^*},f_2\rangle\,.\]
    Thus the $\Psi_{\gamma,\gamma^*}$'s constitute a Parseval frame and we have
    \[f(x)=\sum_{\gamma,\gamma^*\in\Z^2} \Psi_{\gamma,\gamma^*}(x)\, \langle \Psi_{\gamma,\gamma^*},f\rangle\]
    in $L^2(\R^2)$. To show convergence in the seminorms of $\mathscr{S}(\R^2)$ when $f\in\mathscr{S}(\R^2)$, we only need the sum to converge in $\mathscr{S}(\R^2)$. This follows from the seminorms of $\Psi_{\gamma,\gamma^*}$ being polynomially bounded in $\gamma,\gamma^*$ and the coefficients having arbitrary polynomial decay. We leave the details to the reader, see also \cite{Thorn2026}.
\end{proof}

\subsection{Sufficient conditions for being a magnetic pseudo-differential operator}

We now present sufficient conditions ensuring that a given operator on  $L^2(\R^2)$ is a magnetic pseudo-differential operator, in the sense of \eqref{pseudo2}; see part of Theorem~3.1 in \cite{CorneanHelfferPurice2024}. More precisely, we show that a suitable decay of the matrix elements of the operator
with respect to a magnetic Gabor frame implies that the associated symbol belongs to
$BC^\infty(\mathbb R^4)$, the space of bounded smooth functions with bounded derivatives
of all orders.

\begin{lemma}\label{pseudo:matrix}
    Consider the tight magnetic Gabor frame $(\Psi_{\gamma,\gamma^*})_{\gamma,\gamma^*\in\Z^2}$ as in Lemma~\ref{pseudo:frame}.  Let $R\colon \mathscr{S}(\R^2)\to L^2(\R^2)$ be a linear  operator such that both $R$ and its adjoint $R^\ast$ map $\mathscr{S}(\R^2)$ into $ L^2(\R^2)$. Assume moreover that for every $n\geq 1$ there exists $C>0$ such that
\begin{equation}\label{pseudo4}
    |\langle \Psi_{\alpha,\alpha^*}, R\Psi_{\beta,\beta^*}\rangle |\leq C \, \langle \alpha-\beta\rangle^{-n}\, \langle \alpha^*-\beta^*\rangle^{-n},\quad \forall \alpha,\beta,\alpha^*,\beta^*\in \Z^2. 
\end{equation}
Then there exists a symbol $r\in BC^\infty(\R^4)$ such that for all $f_1, f_2 \in\mathscr{S}(\mathbb{R}^2)$, it holds that
\begin{equation}\label{pseudo5}
\langle f_1,R f_2\rangle =(2\pi)^{-2} \int_{\R^2}\di\xi\, \int_{\R^4} \di x\, \di y\, \eu^{\iu\xi\cdot (x-y)} \, \eu^{\iu b\phi(x,y)}\, r\left ( \frac{x+y}{2}, \xi\right) \overline{f_1(x)}\, f_2(y),
\end{equation}
\end{lemma}

\begin{proof}
Let $f_1,f_2\in\mathscr{S}(\R^2)$. We have that $R^*f_1\in L^2(\R^2)$ and using Lemma~\ref{pseudo:frame} and Remark~\ref{rem: Gabor}, we get that
\begin{equation}\label{pseudo6}
\begin{aligned}
\langle f_1, R f_2\rangle&= \langle R^*f_1,  f_2\rangle\\
&= \sum_{\beta,\beta^*\in\Z^2} \langle R^*f_1,\Psi_{\beta,\beta^*}\rangle \, \langle \Psi_{\beta,\beta^*},f_2\rangle =\sum_{\beta,\beta^*\in\Z^2} \langle f_1,R\,\Psi_{\beta,\beta^*}\rangle \, \langle \Psi_{\beta,\beta^*},f_2\rangle\\
&=\sum_{\alpha,\beta,\alpha^*,\beta^*\in\Z^2} \langle \Psi_{\alpha,\alpha^*}, R\Psi_{\beta,\beta^*}\rangle \int_{\R^4} \di x\, \di y\, \Psi_{\alpha,\alpha^*}(x)\,  \overline{\Psi_{\beta,\beta^*}(y)} f_2(y) \overline{f_1(x)}\, \\
&=\sum_{\alpha,\beta,\alpha^*,\beta^*\in\Z^2} \langle \Psi_{\alpha,\alpha^*}, R\Psi_{\beta,\beta^*}\rangle  \times 
\\ 
&\quad \times \int_{\R^4} \di u\, \di v\, \Psi_{\alpha,\alpha^*}\left(u+\frac{v}{2}\right)\,  \overline{\Psi_{\beta,\beta^*}\left(u-\frac{v}{2}\right)} f_2\left(u-\frac{v}{2}\right) \overline{f_1\left(u+\frac{v}{2}\right)},\,
\end{aligned}
\end{equation}
where in the last equality we made a linear change of coordinates $u=(x+y)/2$ and $v=x-y$ for later convenience. The distributional kernel of $R$, in the variables $(u,v)$, is 
\begin{equation*}
\begin{aligned}
R(u,v) = \sum_{\alpha,\beta,\alpha^*,\beta^*\in\Z^2} \langle \Psi_{\alpha,\alpha^*}, R\Psi_{\beta,\beta^*}\rangle\, \Psi_{\alpha,\alpha^*}\left(u+\frac{v}{2}\right)\,  \overline{\Psi_{\beta,\beta^*}\left(u-\frac{v}{2}\right)}.
\end{aligned}
\end{equation*}
We want to define an approximate magnetic symbol from the above integral kernel. To do that, we only keep finitely many terms in the series above. Moreover, it is convenient to multiply the sum by $\eu^{\iu b\phi(u,v)} = \eu^{-\iu b\phi(x,y)}$ and to take a partial inverse Fourier transform with respect to $v$. We denote by $\sum^N$ the truncated series where all indices are, in absolute value, less than $N>1$. Then the approximate symbol is:
\begin{equation}\label{pseudo7}
\begin{aligned}
r_N(u,\xi)\coloneq\sum^N\langle \Psi_{\alpha,\alpha^*}, R\Psi_{\beta,\beta^*}\rangle  \int_{\R^2} \di v\, \eu^{\iu b\phi(u,v)}\, \eu^{-\iu\xi\cdot v} \, \Psi_{\alpha,\alpha^*}\left(u+\frac{v}{2}\right)\,  \overline{\Psi_{\beta,\beta^*}\left(u-\frac{v}{2}\right)}.
\end{aligned}
\end{equation}
The main idea is to rearrange this expression so that it becomes clear that we can take $N\to \infty$ and $r_N$ converges uniformly on compact sets, together with all its derivatives. This limit will be our symbol.

From \eqref{pseudo4} we see that we have any polynomial decay we want in the differences $\alpha-\beta$ and $\alpha^*-\beta^*$, so we only need to control the series in e.g. the sums $\alpha+\beta$ and $\alpha^*+\beta^*$. Let us work on a single term:
\begin{equation*}
\begin{aligned}
    \int_{\R^2} \di v\, \eu^{\iu b\phi(u,v)}\, &\eu^{-\iu\xi\cdot v} \, \Psi_{\alpha,\alpha^*}\left(u+\frac{v}{2}\right)\,  \overline{\Psi_{\beta,\beta^*}\left(u-\frac{v}{2}\right)}\\
    &=\frac{\eu^{\iu(\beta\cdot\beta^*-\alpha\cdot\alpha^*)}\eu^{\iu u\cdot (b(\alpha-\beta)^\perp+\alpha^*-\beta^*)}}{(2\pi)^2}\times \\
    &\quad\times\int_{\R^2}\di v\,\eu^{\iu v\cdot(-\xi-bu^\perp+(b\alpha^\perp+b\beta^\perp+\alpha^*+\beta^*)/2)}g\left(u+\frac{v}{2} -\alpha\right)g\left(u-\frac{v}{2} -\beta\right)
\end{aligned}
\end{equation*}
Using the support properties of the function $g$ (see \eqref{eq:prop g}), we find that the integrand has support in $\alpha-\beta+(-2,2)^2$. By translating the integrand by $\beta-\alpha$ we get:
\begin{align}\label{pseudo8}
    &\frac{\eu^{\iu(\beta\cdot\beta^*-\alpha\cdot\alpha^*)}\eu^{\iu u\cdot(b(\alpha-\beta)^\perp+\alpha^*-\beta^*)}}{(2\pi)^2}\times\\
    &\times\int_{(-2,2)^2}\di w\,\eu^{\iu (w+\alpha-\beta)\cdot({-}\xi-bu^\perp+(b\alpha^\perp+b\beta^\perp+\alpha^*+\beta^*)/2)}g\left(u+\frac{w-\alpha-\beta}{2}\right)g\left(u-\frac{w+\alpha+\beta}{2}\right)\nonumber
\end{align}
Using again \eqref{eq:prop g}, we notice that if $u\notin (\alpha+\beta)/2+(-1,1)^2$, then the above is zero. Hence restricting $u$ to a compact set means that only finitely many values of the sum $\alpha+\beta$ gives non-zero terms in \eqref{pseudo7}, giving us control of the sum in $\alpha+\beta$.

Next using partial integration in $w$ we can get factors of $\langle\xi+(\alpha^*+\beta^*)/2\rangle^{-n}$ for arbitrary~$n$ at the expense of factors of $\langle u-(\alpha+\beta)/2\rangle^n$. The factors of $\langle\xi+(\alpha^*+\beta^*)/2\rangle^{-n}$ gives us control of $\alpha^*+\beta^*$ if $\xi$ is restricted to a compact set. The other factors are bounded since we already found that $u- (\alpha+\beta)/2\notin(-1,1)^2$ implies that \eqref{pseudo8} is equal to zero.

We conclude that $r_N$ converges uniformly on compact sets, leaving us only to study the derivatives of $r_N$. Taking derivatives in $u$ and $\xi$ in \eqref{pseudo8} only results in factors growing like $\langle w\rangle^n\langle\alpha-\beta\rangle^n\langle\alpha^*-\beta^*\rangle^n$ for some natural number $n$, along with derivatives of $g$. But this poses no problems in terms of convergence since $g$ is compactly supported and smooth, the integral in $w$ is over a compact set, and \eqref{pseudo4} provides us with arbitrary off-diagonal decay in $\alpha-\beta$ and $\alpha^*-\beta^*$.
\end{proof}

\begin{lemma}\label{pseudo:translation_invariant}
    Assume that the operator $R$ in \eqref{pseudo5} commutes with all magnetic translations. Then the symbol $r$ does not depend on $u$ and $R={\rm Op}^A(r)$ with $r\in S^0_0(\R^2)$. 
\end{lemma}

\begin{proof}
Recall the magnetic translations $(\tau_z f)(x)=\eu^{\iu b\phi(x,z)}f(x-z)$ with $\phi$ defined in \eqref{pierls} and such that $\tau_z^*=\tau_{-z}$. Let $f_{1},f_2\in\mathscr{S}(\R^2)$, then for all $z\in \R^2$ we have that
\[\langle  f_1,R f_2 \rangle = \langle \tau_z f_1,R\tau_z f_2\rangle.\]
By writing the two quantities above via \eqref{pseudo5}, making the change of variables $x-z\mapsto x$ and $y-z\mapsto y$ in the expression for $\langle \tau_z f_1,R\tau_z f_2\rangle$ and using the identities
\[ -\phi(x+z,z)+\phi(x+z,y+z) +\phi(y+z,z)=\phi(x,y)=\phi(v,u),\]
where $u=(x+y)/2$ and $v=x-y$, we get that for all $z\in \R^2$:
\begin{equation}\label{pseudo9}
0=\int_{\R^2}\di\xi\, \int_{\R^4} \di u\, \di v\, \eu^{\iu\xi\cdot v} \, \eu^{\iu b\phi(v,u)}\, \Big (r\big ( u+z, \xi\big ) -r\big ( u, \xi\big )\Big ) \overline{f_1 \left(u+\frac{v}{2}\right)}\, f_2\left(u-\frac{v}{2}\right).
\end{equation}
 Now for any Schwartz function $F$ in $\R^4$, and using the frame from Lemma~\ref{pseudo:frame}, we may write 
\[F(x,y)=\sum_{\alpha,\alpha^*\in\Z^2}\langle F(x,\cdot),\Psi_{\alpha,\alpha^*}\rangle \,\Psi_{\alpha,\alpha^*}(y),\quad x,y\in \R^2, \]
where the series converges in all Schwartz seminorms. Using this in \eqref{pseudo9} and taking $f_1(x)=\langle F(x,\cdot),\Psi_{\alpha,\alpha^*}\rangle$ and $f_2(y)=\Psi_{\alpha,\alpha^*}(y)$ we have
\[
0=\int_{\R^2}\di\xi\, \int_{\R^4} \di u\, \di v\, \eu^{\iu\xi\cdot v} \, \eu^{\iu\phi(v,u)}\, \Big (r\big ( u+z, \xi\big ) -r\big ( u, \xi\big )\Big ) F\left(u+\frac{v}{2},u-\frac{v}{2}\right)
\]
for all Schwartz functions $F$ in $\R^4$. The map 
\[\mathscr{S}(\R^4)\ni F(u,v)\mapsto G(u,\xi)\coloneq\int_{\R^2}\di v\, \eu^{i\xi\cdot v}\eu^{\iu b\phi(v,u)} F\left(u+\frac{v}{2},u-\frac{v}{2}\right)\in \mathscr{S}(\R^4)\]
is invertible (being the composition of a linear change of coordinates, multiplication by a pointwise invertible, polynomially bounded, smooth function, and an inverse Fourier transform), and so this leads to: 
\[0=\int_{\R^4} \di u\, \di\xi\, \Big (r\big ( u+z, \xi\big ) -r\big ( u, \xi\big )\Big ) \, G(u,\xi) \] 
for every Schwartz function $G$ in $\R^4$. This implies that $r$ does not depend on $u$ and thus, $R= {\rm Op}^A(r)$ and $r\in S_0^0(\mathbb{R}^2)$ from Lemma~\ref{pseudo:matrix}.
\end{proof}

Now we formulate the ``genuine" Beals commutator criterion for the special case in which we deal with constant magnetic fields and operators commuting with all magnetic translations (for another version, see Theorem~1.2 \cite{CorneanHelfferPurice2018}). We denote by $X_j$ multiplication with $x_j$, and let $H_b$ be the operator as in \eqref{eq: def Hb}. We define  
\begin{equation}\label{eq: def Pij op}
    \Pi_j\coloneq\frac{1}{2}\,\iu[H_b,X_j]=(-\iu\partial_{x_j}-bA_j(x)) \qquad j\in\{1,2\}
\end{equation}
with $A$ as in \eqref{eq: def Hb}, i.e., $A(x)=(1/2)(-x_2,x_1)$. They are operators which leave $\mathscr{S}(\R^2)$ invariant. Denote by $T_j$ with $j\geq 1$ any element of the set $\{X_1,X_2,\Pi_1,\Pi_2\}$. 
\begin{corollary}\label{corollary-beals}
    Let $R\colon L^2(\R^2)\rightarrow L^2(\R^2)$ be a bounded operator which commutes with all magnetic translations, and let us assume that all the multiple commutators
    \[[T_1,[T_2,\dots[T_n,R]\dots]],\quad n\geq 1,\]
    which initially are seen as tempered distributions, can be extended to bounded operators on $L^2(\R^2)$. Then there exists $r\in S_0^0(\R^2)$ such that $R={\rm Op}^A(r)$.
\end{corollary}

\begin{proof}
The idea is to show that we may apply Lemma~\ref{pseudo:matrix} in order to construct some $r(u,\xi)$ as in \eqref{pseudo5}; then the invariance to magnetic translations and Lemma~\ref{pseudo:translation_invariant} would show that $r$ only depends on $\xi$. Thus the only thing we need to check is that \eqref{pseudo4} is satisfied. We write
\begin{align*}
    (\alpha_j-\beta_j) \langle \Psi_{\alpha,\alpha^*}, R\Psi_{\beta,\beta^*}\rangle &=\langle X_j\, \Psi_{\alpha,\alpha^*}, R\Psi_{\beta,\beta^*}\rangle + \langle (\alpha_j \mathrm{Id} - X_j)\, \Psi_{\alpha,\alpha^*}, R\Psi_{\beta,\beta^*}\rangle
    \\
    &\quad  - \langle  \Psi_{\alpha,\alpha^*}, R\, X_j\Psi_{\beta,\beta^*}\rangle - \langle \Psi_{\alpha,\alpha^*}, R(\beta_j \mathrm{Id} - X_j)\Psi_{\beta,\beta^*}\rangle  \\
    &=\langle\, \Psi_{\alpha,\alpha^*}, [X_j,R]\Psi_{\beta,\beta^*}\rangle  + \langle (\alpha_j \mathrm{Id} - X_j)\, \Psi_{\alpha,\alpha^*}, R\Psi_{\beta,\beta^*}\rangle
    \\
    &\quad- \langle \Psi_{\alpha,\alpha^*}, R(\beta_j \mathrm{Id} - X_j)\Psi_{\beta,\beta^*}\rangle ,
\end{align*}
Since the function $g$ defining the Gabor frame is compactly supported, the last two terms above, i.e., $\langle (\alpha_j \mathrm{Id} - X_j)\, \Psi_{\alpha,\alpha^*}, R\Psi_{\beta,\beta^*}\rangle - \langle \Psi_{\alpha,\alpha^*}, R(\beta_j \mathrm{Id} - X_j)\Psi_{\beta,\beta^*}\rangle$ are bounded using the compact support of $(x_j-\gamma_j)g(x-\gamma)$. The first term instead, $\langle\, \Psi_{\alpha,\alpha^*}, [X_j,R]\Psi_{\beta,\beta^*}\rangle$, is bounded because the commutator  $[X_j,R]$ can be extended to a bounded operator by assumptions. Hence $(\alpha_j-\beta_j) \langle \Psi_{\alpha,\alpha^*}, R\Psi_{\beta,\beta^*}\rangle$ is uniformly bounded. By repeating the argument we obtain that $(\alpha_j-\beta_j)^n \langle \Psi_{\alpha,\alpha^*}, R\Psi_{\beta,\beta^*}\rangle$ is uniformly bounded for all $n\geq 1$.

Using that $(-\iu\partial_{x_j}-bA_j(x))\eu^{\iu b\phi(x,\gamma)}=\eu^{\iu b\phi(x,\gamma)}\big (-\iu\partial_{x_j}-bA_j(x-\gamma)\big )$, one can similarly show that
\[(\alpha_j^*-\beta_j^*) \langle \Psi_{\alpha,\alpha^*}, R\Psi_{\beta,\beta^*}\rangle =\langle \Pi_j\, \Psi_{\alpha,\alpha^*}, R\Psi_{\beta,\beta^*}\rangle -\langle  \Psi_{\alpha,\alpha^*}, R\, \Pi_j\Psi_{\beta,\beta^*}\rangle +{\rm bounded},\]
uniformly in the indices $\alpha,\beta,\alpha^*,\beta^*$. Since $[\Pi_j,R]$  extends to a bounded operator, $(\alpha_j^*-\beta_j^*) \langle \Psi_{\alpha,\alpha^*}, R\Psi_{\beta,\beta^*}\rangle$ is uniformly bounded. By repeating the argument we obtain that $(\alpha_j^*-\beta_j^*)^n \langle \Psi_{\alpha,\alpha^*}, R\Psi_{\beta,\beta^*}\rangle$ is uniformly bounded for all $n\geq 1$.
\end{proof}

\subsection{Criterion for boundedness}

As a criterion for boundedness, we present a frame-based proof of a simplified magnetic Calderón-Vaillancourt theorem (for the general case see Theorem~3.7 \cite{CorneanHelfferPurice2024}).

\begin{lemma}\label{pseudo:cald-vaill}
    For $p\in S_0^0(\R^2)$, the operator ${\rm Op}^A(p)$ can be extended (by continuity in the strong topology) to a bounded operator on $L^2(\R^2)$. 
\end{lemma}

\begin{proof}
    The adjoint of ${\rm Op}^A(p)$ equals ${\rm Op}^A(\overline{p})$, which sends $\mathscr{S}(\R^2)$ into $L^2(\R^2)$. Hence if $f_1$ and $f_2$ are Schwartz functions we may use the frame from Lemma~\ref{pseudo:frame} and \eqref{pseudo6} {to write}:
    \begin{equation*}
    \begin{aligned}
    \langle f_1, {\rm Op}^A(p) f_2\rangle
    &=\sum_{\alpha,\beta,\alpha^*,\beta^*\in\Z^2} \langle f_1,\Psi_{\alpha,\alpha^*}\rangle\, \langle \Psi_{\alpha,\alpha^*}, {\rm Op}^A(p)\Psi_{\beta,\beta^*}\rangle\, \langle\Psi_{\beta,\beta^*} f_2\rangle\,.
    \end{aligned}
    \end{equation*}
    If we can prove that the matrix $\langle \Psi_{\alpha,\alpha^*}, {\rm Op}^A(p)\Psi_{\beta,\beta^*}\rangle$ is ``sufficiently nice'', a use of the Schur test in $\ell^2(\Z^4)$ and the fact that the frame coefficients obey a Parseval identity, would end the proof. ``Sufficiently nice'' will in this instance mean sufficient off-diagonal decay, i.e. that condition \eqref{pseudo4} holds.
    
    We compute
    \begin{align*} 
    &\langle \Psi_{\alpha,\alpha^*}, {\rm Op}^A(p)\, \Psi_{\beta,\beta^*}\rangle \,  \eu^{\iu \beta\cdot \beta^*-\iu \alpha\cdot \alpha^*} \\
    &= (2\pi)^{-2}\int_{\R^2} \di\xi \int_{\R^4} \di x\, \di y\, \eu^{\iu b\phi(\alpha,x)}g(x-\alpha) \eu^{-\iu \alpha^* \cdot x} \, \eu^{\iu b\phi(x,y)}\, \eu^{\iu\xi\cdot (x-y)} \, p(\xi)\, \eu^{\iu b\phi(y,\beta)}g(y-\beta) \eu^{\iu \beta^* \cdot y}.
    \end{align*}
    Making the change of variables:
    \begin{align*}
        \xi- (\alpha^*+\beta^*)/2&\mapsto \zeta\\
        (x+y-\alpha-\beta)/2&\mapsto u\\
        x-y+\beta-\alpha&\mapsto v
    \end{align*}
    one gets
    \begin{align*} 
    &\langle \Psi_{\alpha,\alpha^*}, {\rm Op}^A(p)\, \Psi_{\beta,\beta^*}\rangle \, \eu^{\iu \beta\cdot \beta^*-\iu \alpha\cdot \alpha^*} \\
    &= (2\pi)^{-2}\eu^{\iu(b\alpha\cdot \beta^\perp+\alpha\cdot\beta^*/2-\beta\cdot\alpha^*/2)}\int_{\R^2} \di\zeta \int_{\R^4} \di u\, \di v\, \eu^{\iu[\zeta\cdot(v+\alpha-\beta)-u\cdot v^\perp+u\cdot(2 b\, (\beta-\alpha)^\perp+\beta^*-\alpha^*)]} \\
    &\quad \quad \times p(\zeta +(\alpha^*+\beta^*)/2)\,g(u+v/2) \,g(u-v/2) .
    \end{align*}
    The function $p(\zeta +(\alpha^*+\beta^*)/2)\,g(u+v/2) \,g(u-v/2) $ has compact support in $u$ and $v$ and has continuous and bounded derivatives in all variables. Using partial integration with respect to $v$, and remembering that the integration in $u$ and $v$ takes place on a compact set, we can create a factor dominated by $\langle\zeta\rangle^{-3}$ which decays even faster if we differentiate it again with respect to $\zeta$. By performing repeated partial integration with respect to $\zeta$, we create an arbitrary decay in  $\langle\alpha-\beta\rangle$, at the price of powers in $v$ which remain bounded on the domain of integration. Finally, partial integration with respect to $u$ gives us an arbitrary decay in $\langle\alpha^*-\beta^*\rangle$ at the price of grow in $|b|\, \langle\alpha-\beta\rangle,$ but this growth can be absorbed by extra partial integration in $\zeta$. Hence \eqref{pseudo4} is satisfied.
\end{proof}

\subsection{Composition of magnetic pseudo-differential operators}

When it comes to composition of magnetic pseudo-differential operators one works with the magnetic Moyal product, see Section~V \cite{Mantoiu2004} and Corollary~3.6 \cite{CorneanHelfferPurice2024}, of which we present a specific case.

\begin{lemma}\label{pseudo:moyal}
If $p\in S_0^{-m}(\R^2)$ and $q\in S_0^{-n}(\R^2)$ with $m,n\geq 0$, then there exists $t\in S_0^{-n-m}(\R^2)$ such that as bounded operators on $L^2(\R^2)$ we have: \[{\rm Op}^A(t)={\rm Op}^A(p)\, {\rm Op}^A(q).\] 
\end{lemma}

\begin{proof}
For $\epsilon>0$ we denote by $p_\epsilon$ and $q_\epsilon$ the regularization of $p$ and $q$ with respect to $\epsilon$, i.e., 
\begin{equation}\label{eq: def peps qeps}
    p_\epsilon(\xi)=p(\xi)\eu^{-\epsilon\langle \xi\rangle}, \qquad  q_\epsilon(\xi)=q(\xi)\eu^{-\epsilon\langle \xi\rangle}
\end{equation}
The corresponding operators ${\rm Op}^A(p_\epsilon)$ and ${\rm Op}^A(q_\epsilon)$  have integral kernels given by:
\begin{equation*}
\begin{aligned}
\eu^{\iu b\phi(x,y)} \widecheck{p_\epsilon}(x-y)&\coloneq\eu^{\iu b\phi(x,y)} (2\pi)^{-2} \int_{\R^2}\di\xi\, p_\epsilon(\xi)\eu^{ \iu\xi\cdot (x-y)}.
\end{aligned}
\end{equation*}
\begin{equation*}
\begin{aligned}
\eu^{\iu b\phi(x,y)} \widecheck{q_\epsilon}(x-y)&\coloneq\eu^{\iu b\phi(x,y)} (2\pi)^{-2} \int_{\R^2}\di\xi\, q_\epsilon(\xi)\eu^{ \iu\xi\cdot (x-y)}.
\end{aligned}
\end{equation*}
Thus the integral kernel of ${\rm Op}^A(p_\epsilon){\rm Op}^A(q_\epsilon)$ is given by:
\[
\int_{\R^2}\di y\, \eu^{\iu b(\phi(x,y) +\phi(y,x'))}\widecheck{p_\epsilon}(x-y) \widecheck{q_\epsilon}(y-x')
\]
Using Proposition~\ref{prop:MagCommuteKernel} and the fact that ${\rm Op}^A(p_\epsilon){\rm Op}^A(q_\epsilon)$ commutes with all the magnetic translations, its kernel equals $\eu^{\iu b\phi(x,x')} T_\epsilon(x-x')$ with 
\[
T_\epsilon(x)\coloneq\int_{\R^2}\di y\, \eu^{\iu b\phi(x,y)}\widecheck{p_\epsilon}(x-y) \widecheck{q_\epsilon}(y).
\]
This is a Schwartz function and so we see that ${\rm Op}^A(p_\epsilon){\rm Op}^A(q_\epsilon)$ is a magnetic pseudo-differential operator with symbol:
\begin{equation}\label{pseudo10}
\begin{aligned}
t_\epsilon(\xi)&=\int_{\R^2}\di x\, \eu^{-\iu\xi \cdot x} T_\epsilon(x)  =\int_{\R^2} \di y\, \int_{\R^2} \di x\,\eu^{-\iu(\xi+by^\perp) \cdot (x-y)}\widecheck{p_\epsilon}(x-y)\,  \eu^{-\iu\xi\cdot y} \widecheck{q_\epsilon}(y)  \\
&= \int_{\R^2}\di y\, p_\epsilon(\xi +by^\perp) \widecheck{q_\epsilon}(y) \eu^{-\iu\xi \cdot y}.
\end{aligned}
\end{equation}
Now the goal is to show that $t_\epsilon$ converges suitably to an element of $S_0^{-m-n}(\R^2)$ when $\epsilon\rightarrow0$. Consider a tight magnetic Gabor frame $(\Psi_{\gamma,\gamma^*})_{\gamma,\gamma^*\in\Z^2}$ from Lemma~\ref{pseudo:frame}.
Since $\widecheck{q_\epsilon}(y)$ is a Schwartz function we have a ``pointwise Parseval decomposition''
\begin{equation*}
\begin{aligned}
\widecheck{q_\epsilon}(y)&=\sum_{\gamma,\gamma^*\in\Z^2}\Psi_{\gamma,\gamma^*}(y)\langle\Psi_{\gamma,\gamma^*},\widecheck{q_\epsilon}\rangle\\
&=\sum_{\gamma,\gamma^*\in\Z^2} (2\pi)^{-2} g(y-\gamma)\eu^{\iu (\gamma^*+b\gamma^\perp) \cdot (y-\gamma)}\int_{\R^2}\di w\, g(w-\gamma)\eu^{-\iu (\gamma^*+b\gamma^\perp) \cdot (w-\gamma)}\widecheck{q_\epsilon}(w)
\end{aligned}
\end{equation*}
and using the change of variables $y-\gamma \mapsto y$ and $w-\gamma\mapsto w$ we have
\begin{align*}
  & t_\epsilon(\xi)=\sum_{\gamma,\gamma^*\in\Z^2} (2\pi)^{-2} \int_{\R^2}\di y\, p_\epsilon(\xi +y^\perp) \eu^{- \iu\xi \cdot y} g(y-\gamma)\eu^{\iu (\gamma^*+b\gamma^\perp) \cdot (y-\gamma)}\\
  &\quad\quad\quad\quad\times\int_{\R^2}\di w\, g(w-\gamma)\eu^{-\iu (\gamma^*+b\gamma^\perp) \cdot (w-\gamma)}\widecheck{q_\epsilon}(w) \\ 
  &=\sum_{\gamma,\gamma^*\in\Z^2} (2\pi)^{-2} \int_{\R^2}\di y\, p_\epsilon(\xi +\gamma^\perp +y^\perp) \eu^{-\iu\xi \cdot (y+\gamma)} g(y)\eu^{\iu (\gamma^*+b\gamma^\perp) \cdot y}\\
  &\quad\quad\quad\quad\times\int_{\R^2}\di w\, g(w)\eu^{-\iu (\gamma^*+b\gamma^\perp) \cdot w}\widecheck{q_\epsilon}(w+\gamma) 
  \\ 
  &=\sum_{\gamma,\gamma^*\in\Z^2} (2\pi)^{-2} \int_{\R^2} \di y\,\eu^{\iu(\gamma^*+b\gamma^\perp-\xi) \cdot y-\iu\xi\cdot \gamma}p_\epsilon(\xi +\gamma^\perp +y^\perp) g(y)\\
  &\quad\quad\quad\quad \times\int_{\R^2}\di w\, g(w)\eu^{-\iu (\gamma^*+b\gamma^\perp) \cdot w}\widecheck{q_\epsilon}(w+\gamma).
\end{align*}
We first show that $t_\epsilon\in S_0^0$ uniformly in $\epsilon$, which would imply that $t_\epsilon$ and all its derivatives have a uniform limit on compact sets. Secondly, we prove that if $p\in S_0^{-m}(\R^2)$ and $q\in S_0^{-n}(\R^2)$, then the limit $t$ of $t_\epsilon$ belongs to $S_0^{-n-m}(\R^2)$. Lastly, we verify that ${\rm Op}^A(t)={\rm Op}^A(p)\, {\rm Op}^A(q)$.

\noindent{\bf Step 1}. The integral with respect to $w$ in the above last equality does not depend on $\xi$. We have:
\begin{align*}
  &\int_{\R^2}\di w\, g(w)\eu^{-\iu (\gamma^*+b\gamma^\perp) \cdot w}\widecheck{q_\epsilon}(w+\gamma) =(2\pi)^{-1} \int_{\R^2} \di\eta \, q_\epsilon(\eta) \, \eu^{\iu\eta\cdot \gamma} \int_{\R^2} \di w \, g(w) \eu^{\iu(\eta-\gamma^*-b\gamma^\perp)\cdot w}
  \\
 &=\int_{\R^2} q_\epsilon(\eta) \eu^{\iu\eta\cdot \gamma} \widecheck{g}(\eta-\gamma^*-b\gamma^\perp)\di\eta =\int_{\R^2}  \eu^{\iu(\eta+\gamma^*)\cdot \gamma}\, q_\epsilon(\eta+\gamma^*+b\gamma^\perp) \,\widecheck{g}(\eta)\di\eta 
\end{align*}
which by partial integration in $\eta$ provides rapid decay in $\gamma$ uniformly in $\epsilon$ and $\gamma^*$, using the definition of $q_\epsilon$ in \eqref{eq: def peps qeps} and that $q\in S_0^{-n}(\mathbb{R}^2)$ together with $g\in C^\infty_0(\mathbb{R}^2)$. The variable $\xi$ only appears in the first integral:
\begin{align*}
&\partial_{\xi_j}^N \int_{\R^2}\di y\, \eu^{\iu(\gamma^*+b\gamma^\perp-\xi) \cdot y-\iu\xi\cdot \gamma}p_\epsilon(\xi +\gamma^\perp +y^\perp) g(y)\\
&\qquad =\sum_{0\leq k,r,s\leq N}C_{k,r,s}\gamma_j^k\int_{\R^2}\di y\, \eu^{\iu(\gamma^*+b\gamma^\perp-\xi) \cdot y-\iu\xi\cdot \gamma}\partial_{\xi_j}^{r}p_\epsilon(\xi +\gamma^\perp +y^\perp) \, y_j^{s} g(y).
\end{align*}
By partial integration with respect to $y$, all integrals decay faster than any power of $\langle \gamma^*+b\gamma^\perp-\xi\rangle$ where we are using that $g\in C^\infty_0$ and $p_\epsilon \in C^\infty$. Together with Peetre's inequality and the previous factors which control any growth in $\gamma$, we insure the absolute convergence of the series in $\gamma^*$, uniformly in $\epsilon$. This guarantees that $t_\varepsilon \in S_0^0$.

\noindent{\bf Step 2.} Now let us show that $t_\epsilon(\xi)$ belongs to $S_0^{-n-m}(\R^2)$ uniformly in $\epsilon$, provided that $p\in S_0^{-m}(\R^2)$ and $q\in S_0^{-n}(\R^2)$.

Using \eqref{pseudo10} we have 
\begin{align*}
\xi_j^{m+n} t_\epsilon(\xi)&= \int_{\R^2}\di y\, \xi_j^m p_\epsilon(\xi +y^\perp) \widecheck{q_\epsilon}(y) \Big ((i\partial_{y_j})^n\eu^{-\iu\xi \cdot y}\Big )\\
&=\sum_{k=0}^n C_{k} \int_{\R^2}\di y\, \xi_j^m p_{\epsilon,k}(\xi +y^\perp) \widecheck{{q}_{\epsilon,k}}(y) \eu^{-\iu\xi \cdot y},
\end{align*}
where ${p}_{\epsilon,k}(\xi + y^\perp)\coloneq\partial_{-\xi^\perp}^{n-k} p_\epsilon(\xi + y^\perp)$ and $q_{\epsilon,k}(\eta)\coloneq\eta_j^k q_\epsilon(\eta)$. We write $\xi_j^m=(\xi_j+y^\perp_j-y_j^\perp)^m$ and 
\begin{align}\label{pseudo11}
\xi_j^{m+n}  t_\epsilon(\xi)
&=\sum_{k=0}^n\sum_{r=0}^m C_{k,r} \int_{\R^2}\di y\,  p_{\epsilon,k,r}(\xi +y^\perp) \widecheck{{q}_{\epsilon,k,r}}(y) \eu^{-\iu\xi \cdot y},
\end{align}
where $p_{\epsilon,k,r}(\xi)\coloneq\xi_j^rp_{\epsilon,k}(\xi)$ and $q_{\epsilon,k,r}(\eta)\coloneq\partial_{-\eta_j^\perp}^{m-r}\eta_j^k  q_\epsilon(\eta)$. All these symbols belong to $S_0^0(\R^2)$, uniformly in $\epsilon$. Now consider a smooth partition of unity in $\R^2$ given by \[f_0(\xi)+f_1(\xi)+f_2(\xi)=1,\]
where $f_0$ is supported in $\{|\xi|\leq 10\}$, $f_1$ is supported in $\{|\xi|\geq 1 \, \&\, |\xi_1|\geq  |\xi_2|/2\}$  and 
$f_2$ is supported in $\{|\xi|\geq 1 \, \&\,|\xi_2|\geq |\xi_1|/2\}$. We see that $\xi_j\neq 0$ on the support of $f_j$, $j \in \{1,2\}$. Then $\xi_j^{-m-n} f_j(\xi)$ are smooth and decay like $\langle \xi\rangle^{-n-m}$ together with all their derivatives. By writing 
\[t_\epsilon(\xi)=f_0(\xi)\, t_\epsilon(\xi)+\sum_{j=1}^2 \xi_j^{-m-n} f_j(\xi) \Big (\xi_j^{m+n} t_\epsilon(\xi)\Big )\]
and applying Step 1 to each term in \eqref{pseudo11}, we conclude that $\langle \xi\rangle^{n+m}t_\epsilon$ will have a uniform limit on compacts when $\epsilon\to 0$. In a similar way one can prove that $\langle \xi\rangle^{n+m}\partial^\alpha t_\epsilon$ has a uniform limit on compacts when $\epsilon\to 0$. We can then conclude that the limit $t$ of $t_\epsilon$ as $\epsilon \rightarrow 0$ is such that $t\in S_0^{-m-n}(\mathbb{R}^2)$.

\noindent{\bf Step 3.} Let $t\in S_0^{-n-m}(\R^2)$ denote the limit of $t_\epsilon$. 

The convergence of $t_\epsilon\rightarrow t$ on compacts, along with the derivatives, implies that ${\rm Op}^A(t_\epsilon)f$ converges to ${\rm Op}^A(t)f$ in $L^2(\R^2)$-norm for every Schwartz function $f$. Similar conclusions hold for $p_\epsilon,p$ and $q_\epsilon,q$. Then for Schwartz functions $f_1,f_2$, we can conclude that
\[\langle {\rm Op}^A(t)f_1,f_2\rangle=\langle {\rm Op}^A(p){\rm Op}^A(q)f_1,f_2\rangle.\]
Since ${\rm Op}^A(t)$ and ${\rm Op}^A(p){\rm Op}^A(q)$ are both bounded, they must then be equal.
\end{proof}

\subsection{Resolvents are also magnetic pseudo-differential operators}\label{subsec:resolvents}

The resolvent of the free Hamiltonian $-\Delta$ without a magnetic field, i.e. $b=0$, is a pseudo-differential operator with an explicit symbol, e.g. $(-\Delta+1)^{-1}$ has symbol $\langle\xi\rangle^{-2}\in S_0^{-2}(\R^2)$. We can also say something when a magnetic field is present:

\begin{lemma}\label{pseudo:resolvent}
Let $H_b=(-\iu\nabla_x -bA(x))^2$ with $b>0$. Then the resolvent  $(H_b+1)^{-1}$ is a magnetic pseudo-differential operator with a symbol of class $S_0^{-2}(\R^2)$.
\end{lemma}

\begin{proof} 
\noindent{\bf Step 1.} We will first show that \eqref{pseudo4} is satisfied by $R=(H_b+1)^{-1}$, which will imply that $(H_b+1)^{-1}|_{\mathscr{S}(\R^2)}={\rm Op}^A(r)$ for some $r\in S_0^0(\R^2)$. We do not apply Corollary~\ref{corollary-beals} because this would require a rigorous control on the multiple commutators with the resolvent, which is more difficult than proving \eqref{pseudo4} directly. Even this will need a certain regularization procedure explained in what follows.
    
Let $T_j$ be one of the operators $X_1$, $X_2$, $\Pi_1$, or $\Pi_2$, where $X_j$ is the multiplication by $x_j$ and $\Pi_j$ is as defined in \eqref{eq: def Pij op}. We have the commutators (on $\mathscr{S}(\R^2)$):
\[[X_1, \Pi_1]=[X_2,\Pi_2]= \iu,\quad [\Pi_1,\Pi_2]=\iu b,\quad [X_j,H_b]=2\iu\Pi_j,\]
\[[\Pi_1,H_b]=2\iu b\Pi_2,\quad [\Pi_2,H_b]=-2\iu b\Pi_1.\]
We see that $[T_j,H_b]$ is relatively bounded with respect to $H_b$ and, at least formally, we have
\[[T_j,(H_b+1)^{-1}]=-(H_b+1)^{-1}[T_j,H_b](H_b+1)^{-1},\]
where the right hand side is a bounded operator on $L^2(\R^2)$. {We now replace} $X_j$ with $X_j/(1+\iu\epsilon X_j)$ and $\Pi_j$ with $(1+\epsilon H_b)^{-1} \Pi_j (1+\epsilon H_b)^{-1}$ for some $\epsilon>0$, and we see that the new regularized operators $T_{j,\epsilon}$ are bounded, leave the domain of $H_b$ invariant, and converge strongly to $T_j$. Getting back to the proof of Corollary~\ref{corollary-beals}, and denoting $(H_b+1)^{-1}$ by $R$, we have:
\[(\alpha_j-\beta_j) \langle \Psi_{\alpha,\alpha^*}, R\Psi_{\beta,\beta^*}\rangle =\langle X_j\, \Psi_{\alpha,\alpha^*}, R\Psi_{\beta,\beta^*}\rangle -\langle  \Psi_{\alpha,\alpha^*}, R\, X_j\Psi_{\beta,\beta^*}\rangle +{\rm bounded},\]
uniformly in $\alpha,\beta,\alpha^*,\beta^*$. At this moment we may replace $X_j$ with $X_j/(1+\iu\epsilon X_j)$ and perform the commutator with $R$. This commutator will remain uniformly bounded when $\epsilon\to 0$. The same happens for higher order commutators, and also when we consider $\Pi_{j,\epsilon}$ in order to control the decay in $\alpha^*-\beta^*$.  
    
\noindent{\bf Step 2.}  We will now show by a boot-strap argument that the symbol $r$ belongs to $S_0^{-2}(\R^2)$. 
For $f_1,f_2\in\mathscr{S}(\R^2)$, and using that $(H_b+1)(H_b+1)^{-1}$ is the identity operator on $L^2(\R^2)$ we have
\begin{align*}
    \langle f_1,f_2\rangle_{\mathscr{S}'(\R^2),\mathscr{S}(\R^2)}&=\langle (H_b+1){\rm Op}^A(r)f_1,f_2\rangle_{\mathscr{S}'(\R^2),\mathscr{S}(\R^2)}\\
    &=\langle {\rm Op}^A(r)f_1,[(\iu\nabla_x-bA(x))^2+1]f_2\rangle_{\mathscr{S}'(\R^2),\mathscr{S}(\R^2)}.
\end{align*}
By partial integration this identity becomes
\begin{align*}
    \langle f_1,f_2\rangle_{\mathscr{S}'(\R^2),\mathscr{S}(\R^2)}&=\left\langle {\rm Op}^A\left(\left[\left(2A(\xi)+\iu\frac{b}{2}\nabla_\xi\right)^2+1\right] r\right)f_1,f_2\right\rangle_{\mathscr{S}'(\R^2),\mathscr{S}(\R^2)}.
\end{align*}
The identity operator on $\mathscr{S}(\R^2)$ has symbol $1\in S_0^0(\R^2)$, hence by a similar argument as in Lemma~\ref{pseudo:translation_invariant} we obtain:
\[1=\left[\left(2A(\xi)+\iu\frac{b}{2}\nabla_\xi\right)^2+1\right] r(\xi).\]
This is the same as: 
\[r(\xi)=\frac{1}{|\xi|^2+1}\left(b\iu\xi_1\partial_{\xi_2}r(\xi)-b\iu\xi_2\partial_{\xi_1}r(\xi)-\frac{b^2}{4}\Delta_\xi r(\xi)+1\right),\]
and using the a-priori information that all the derivatives of $r$ are uniformly bounded, we get in the first step that $r\in S_0^{-1}(\R^2)$, where the dominating terms are $\frac{1}{|\xi|^2+1}b\iu\xi_1\partial_{\xi_2}r(\xi)$ and $-\frac{1}{|\xi|^2+1}b\iu\xi_2\partial_{\xi_1}r(\xi)$. Now using in the same identity that $r\in S_0^{-1}(\R^2)$, we finally obtain that $r\in S_0^{-2}(\R^2)$.
\end{proof}

\begin{lemma}\label{pseudo:resolvent2}
Let $W={\rm Op}^A(w)$ with a real symbol $w\in S^0_0(\R^2)$ and consider the self-adjoint operator $H=H_b+W$. Then for all $z\not\in \sigma(H)$, the resolvent $(H-z)^{-1}$ is a magnetic pseudo-differential operator with a symbol $r_z\in S^{-2}_0(\R^2)$. Moreover, for every $\alpha\in \Z^2_+$ there exists $C,N>0$ independent on $z$ such that 
\begin{equation}\label{hdc4}
    |\partial^\alpha r_z(\xi)|\leq C \, \langle \xi\rangle^{-2}\frac{\langle z\rangle^N}{\min\big \{1,\, {\rm dist}(z,\sigma(H))\big\}^N}. 
\end{equation}
\end{lemma}

\begin{proof}
By using partial integration we may compute the following commutators on $\mathscr{S}(\R^2)$: 
\[[X_j,W]=\iu{\rm Op}^A(\partial_{\xi_j}w),\quad [\Pi_1,W]=\iu b{\rm Op}^A(\partial_{\xi_2}w),\quad [\Pi_2,W]=-\iu b{\rm Op}^A(\partial_{\xi_1}w),\]
These commutators can be extended to bounded operators on $L^2(\R^2)$ by Lemma~\ref{pseudo:cald-vaill}. This implies that we may again use the commutator strategy of Lemma~\ref{pseudo:resolvent} in order to control the decay in \eqref{pseudo4} and to construct $r_z$ as an element of $S_0^0(\R^2)$. In order to control the behavior in $z$ we need to bound the operator norms of products of operators like 
\[\Pi_j (H-z)^{-1}=\Pi_j (H+1)^{-1}+(z+1)\Pi_j (H+1)^{-1}(H-z)^{-1}\]
and 
\[\Pi_j (H+1)^{-1}=\Pi_j (H_b+1)^{-1}-\Pi_j (H_b+1)^{-1}W (H+1)^{-1},\]
and use the spectral theorem. Finally, we conclude that $r_z\in S_0^{-2}(\R^2)$ by using again the above resolvent identities together with the symbol composition result of Lemma~\ref{pseudo:moyal}. 
\end{proof}

\section{Magnetic translations and integral kernels}\label{sec:kernels}
We now focus on proving the first result announced in Section~\ref{subsec:main}.

Recall that we are interested in formulating the self-consistency equations~\eqref{eqn:W} and \eqref{eqn:fixed-point-gammastar} directly in the infinite-volume limit, namely
\begin{equation}\label{eq: self consist eq 0}
\begin{aligned}
W(x,x')& = \delta (x-x')\int_{\R^2} v(x-y)\, f (H_b+ \lambda W)(y,y)\di y\\ 
&\qquad - v(x-x')\, f (H_b+ \lambda W)(x,x')\,,\quad x,x'\in \R^2,
\end{aligned}
\end{equation}
where we recall that the self-consistent operator $W$ should be bounded and self-adjoint and that $v(x-y) = V(x;y)$ is the two-body interaction potential such that $\langle x \rangle^n v(x) \in L^1(\mathbb{R}^2)$ for all $n\geq 0$. As already anticipated,  $v$ can be chosen as the screened Coulomb potential in 2D, i.e.\ $v(x)=e^{-\alpha|x|}\,\ln(|x|)$ with $\alpha>0$. 

As explained in Section~\ref{subsec:main}, in order to prove Theorem~\ref{thm1}  we first reformulate the self-consistent equation above. To do that, we needed Proposition~\ref{prop:MagCommuteKernel}, which we now prove.

\begin{proof}[Proof of Proposition~\ref{prop:MagCommuteKernel}]
    Let us first assume that $\tau_{b,y}\,T\,\tau_{b,-y}=T$ for all $y\in\R^2$ (note: $\tau_{b,y}^*=\tau_{b,y}^{-1}=\tau_{b,-y}$). We have
    \begin{align*}
        (\tau_{b,y} \, T \,\tau_{b,-y}\, \psi)(x) &= \int_{\R^2} \eu^{\iu \, b \, \phi(x,y)} \, T(x-y,x') \, \eu^{\iu \, b \, \phi(x',-y)} \,\psi(x'+y) \, \di x'\\
        &= \int_{\R^2} \eu^{\iu \, b \, \phi(x-x',y)} \, T(x-y,x'-y) \, \psi(x') \, \di x'.
    \end{align*}
    Since $\tau_{b,y}\,T\,\tau_{b,-y}=T$, the above readily implies the identity: 
    \[ \eu^{\iu \, b \, \phi(x-x',y)} \, T(x-y,x'-y) = T(x,x') \quad \text{for all } x, y, x' \in \R^2 \,.\]
    The statement follows upon setting $y=x'$.

    Conversely, let us assume that $T(x,x')=\eu^{\iu \, b \, \phi(x,x')} \, F_T(x-x')$, where $F_T$ is locally integrable. Let $\psi\in C_0^\infty(\R^2)$ and compute:
    \begin{align*}
        (\tau_{b,y}\,T\,\tau_{b,-y})\,\psi(x)&=\int_{\R^2} \eu^{\iu \, b \, \phi(x-x',y)} \, \eu^{\iu \, b \, \phi(x-y,x'-y)}\, F_T(x-y-x'+y) \, \psi(x') \, \di x'\\
        &=\int_{\R^2} \eu^{\iu \, b \, \phi(x,x')}\, F_T(x-x') \, \psi(x') \, \di x'=T\,\psi(x)
    \end{align*}
    We conclude that $\tau_{b,y} \, T \, \psi = T \, \tau_{b,y} \, \psi$ for all $\psi\in C_0^\infty(\R^2)$, and by density that $T$ commutes with $\tau_{b,y}$ for all $y \in \R^2$.
\end{proof}

Corollary~\ref{cor:MagCommuteKernel} follows immediately from Proposition~\ref{prop:MagCommuteKernel}.

Assuming that the operators $W$ and $f(H_b+\lambda W)$ commute with all magnetic translations, by using Proposition~\ref{prop:MagCommuteKernel} and Corollary~\ref{cor:MagCommuteKernel}, we can restate \eqref{eq: self consist eq 0} as follows: find a function $F$ such that 
\begin{equation}\label{1sec}
\begin{aligned}
    &\qquad \qquad \qquad \qquad F(x)=f(H_b+\lambda W_{F})(x,0),\\
    W_{F}&= \left( \int_{\R^2} v(y) \, \di y \right) \, F(0) + Z_{F}\, ,\quad Z_{F}(x,x') \coloneq - \eu^{\iu \, b \, \phi(x,x')} \, v(x-x') \, F(x-x') .
    \end{aligned}
\end{equation}
In the following we will assume that $F$ belongs to the space 
\begin{equation}\label{cdh1sec3} BC_H(\R^2)=\{G\in BC(\R^2)\mid \overline{G(-x)} = G(x)\,,\:\forall x\in\R^2\}\,,
\end{equation}
and we write 
\begin{equation}\label{eq: def HblambdaF}
    H_{b,\lambda, F}= H_b + \lambda W_F.
\end{equation}
\subsection{The operators $Z_F$ and $W_F$}
We let $\mathcal{L}(B,B')$ denote the space of continuous linear maps between Banach spaces $B$ and $B'$, and we write $\mathcal{L}(B)$ in place of $\mathcal{L}(B,B)$. We define a map 
\[L^\infty(\R^2)\ni F\mapsto \mathcal{Z}(F) \in \mathcal{L}(L^2(\R^2)),\]
where $\mathcal{Z}(F)$ is an integral operator with integral kernel given by 
\[ \big (\mathcal{Z}(F)\big )(x,x')= - \eu^{\iu \, b \, \phi(x,x')} \, v(x-x') \, F(x-x')\,. \]
We note that if $F\in BC_H(\R^2)$ (see \eqref{cdh1sec3}) then $\mathcal{Z}(F)=Z_F$. 
\begin{lemma}\label{opZ}
The map $\mathcal{Z}$ enjoys the following properties:
\begin{enumerate}[label={\rm(\roman*)}, ref={\rm(\roman*)}]
 \item $\mathcal{Z} : L^\infty(\R^2) \to \mathcal{L}(L^2(\R^2))$ is bounded with operator norm bounded by $\Vert v\Vert_{L^1(\R^2)}$. 
 \item \label{item:opZii} For any $F \in L^{\infty}(\R^2)$, the operator $\mathcal{Z}(F)$ commutes with all magnetic translations~\eqref{eqn:MagTransl}. Moreover, $\mathcal{Z}(F)=\mathcal{Z}(F)^*$ if and only if $F(x)=\overline{F(-x)}$. 
 \item $\mathcal{Z}(F)$ is a magnetic pseudo-differential operator with symbol 
 \[ \zeta_F(\xi) = - \int_{\R^2} \eu^{\iu \xi \cdot x}\, v(x) \, F(x)\, \di x\: \in S^0_0(\R^2)\,. \]
\end{enumerate}
\end{lemma}

\begin{proof}{\color{white}=}

\begin{enumerate}[label={(\roman*)}, ref={(\roman*)}]
 \item The Schur test (Theorem~1.8 in \cite{Hedenmalm2000}) gives
 \begin{align*}
    \Vert \mathcal{Z}(F)\Vert_{\mathcal{L}(L^2(\R^2))} & \leq \left( \sup_{x\in\R^2} \int_{\R^2} \big| \big ( \mathcal{Z}(F)\big )(x,x') \big| \, \di x' \right)^{1/2} \left( \sup_{x'\in\R^2} \int_{\R^2} \big|\big (\mathcal{Z}(F)\big )(x,x')\big|\, \di x \right)^{1/2} \\
    & \leq \Vert v\Vert_{L^1(\R^2)}\, \Vert F\Vert_{L^\infty(\R^2)}\,.
 \end{align*}
 \item The statement follows immediately from Proposition~\ref{prop:MagCommuteKernel} and Corollary~\ref{cor:MagCommuteKernel}.
 \item We see that all the derivatives of $\zeta_F$ are uniformly bounded due to the fast decay of $v$ and the fact that $F$ is bounded. The symbol neither has, nor gains decay in $\xi$ after differentiation, because $v\, F$ is not necessarily continuous anywhere. \qedhere
\end{enumerate}
\end{proof}

In view of the definition~\eqref{1} of $W_F$, we have the following result:
\begin{corollary} \label{cor:WF}
Let $F \in BC_H(\R^2)$.
\begin{enumerate}[label={\rm(\roman*)}, ref={\rm(\roman*)}]
 \item The operator $W_F$ is bounded on $L^2(\R^2)$ with
 \[ \Vert W_F \Vert_{L^2(\R^2)} \le \Vert v \Vert_{L^1(\R^2)} \, \left( |F(0)| + \Vert F \Vert_{L^\infty(\R^2)} \right) \le 2 \,\Vert v \Vert_{L^1(\R^2)} \, \Vert F \Vert_{L^\infty(\R^2)}\,. \] 
 \item $W_F$ is self-adjoint and commutes with all magnetic translations~\eqref{eqn:MagTransl}.
 \item $W_F$ is a magnetic pseudo-differential operator with symbol 
 \[ w_F(\xi) = \left(\int_{\R^2} v(y) \, \di y \right) F(0) - \int_{\R^2} \eu^{\iu \xi \cdot x}\, v(x) \, F(x)\, \di x\: \in S^0_0(\R^2)\,. \]
\end{enumerate}
In particular, for all $\lambda \in \R$ the operator $H_{b,\lambda,F}$ defined in \eqref{eqn:HblF} is bounded from below and self-adjoint on the domain of $H_b$. 
\end{corollary}

\section{The fixed point}\label{sec:fixed}

\subsection{Well-posedness of the fixed-point equation}
We now turn our attention to the fixed-point equation~\eqref{eqn:FixedPoint}. An important technical tool to analyze it, which will be used several times, is the Helffer--Sjöstrand formula, so let us introduce almost analytic extensions and the formula itself. This introduction will be based upon \cite{Davies1995}.

Fix $g\in\mathscr{S}(\R)$ a real valued function and construct $\chi\in C_0^\infty(\R)$ such that $1_{[-1,1]}\leq\chi\leq 1_{[-2,2]}$ holds pointwise. Then for any $N\in\N$ we define
\[\Co\ni z=x+\iu y\mapsto \tilde{g}_N(z)\coloneq\chi\left(y\right)\sum_{n=0}^N\,\frac{\iu^n}{n!}\, g^{(n)}(x)\, y^n\in \Co\, .\]
We call $\tilde{g}_N$ an \emph{almost analytic extension} of $g$. It satisfies (here $\overline{\partial_z}=\frac{1}{2}(\partial_x+\iu\partial_y)$): 

\begin{equation}\label{almost_analytic}
    \begin{aligned}
        \tilde{g}_N|_\R&=g|_\R,\\
        |\tilde{g}_N(z)|&\leq C_m\langle z\rangle^{-m},\quad &\forall m\geq0,\\
        |\overline{\partial_z}\tilde{g}_Z(z)|&\leq C_{m,N}|\Im(z)|^N\langle z\rangle^{-m},\quad&\forall m\geq0.
    \end{aligned}
\end{equation}
Then, for any self-adjoint operator $T$ and for any $N\in\N$, the Helffer--Sjöstrand formula allows us to write:
\[ g(T)=-\frac{1}{\pi}\int_{\R\times [-2,2]}\overline{\partial_z}\tilde{g}_N(z)(T-z)^{-1}\di x\di y,\]
where the integral is understood in the norm topology. If $g$ is complex valued, then we split it into its real and imaginary parts.  

The next proposition is one of the main technical results of our paper. In particular, we show that $\Phi_\lambda$ leaves $BC_H(\R^2)$ invariant. 

\begin{proposition}\label{prop:fermi_dirac_kernel} 
Let $f$ be a real-valued function such that for any $n,k\geq 0$ it satisfies
\begin{equation}\label{cdh2sec4}
\sup_{x\in[0,\infty)}\langle x\rangle^n|f^{(k)}(x)|<\infty.
\end{equation}
and let $H_{b,\lambda,F}$ as in \eqref{eq: def HblambdaF}. It holds that

\begin{enumerate}[label={\rm(\roman*)}, ref={\rm(\roman*)}]
 \item For every $F\in BC_H(\R^2)$, the operator $f(H_{b,\lambda,F})$ is bounded, self-adjoint, and commutes with all magnetic translations. 
 \item \label{item:FDKii} $f(H_{b,\lambda,F})$ is also a magnetic pseudo-differential operator with a smoothing symbol $p_F\in \bigcap_{n\geq0}S_0^{-n}(\R^2)$. 
 \item \label{item:FDKiii} $f(H_{b,\lambda,F})$ has a smooth  integral kernel given by 
 \begin{equation}\label{eq: fhblambdaF} f(H_{b,\lambda,F})(x,x') \coloneq \eu^{\iu \, b \, \phi(x,x')} \, (2\pi)^{-2} \, \int_{\R^2} \eu^{\iu \, \xi \cdot (x-x')}\, p_F(\xi)\, \di\xi\,, \end{equation}
 and setting $x'=0$ we have $f(H_{b,\lambda,F})(\cdot,0)\in BC_H(\R^2)\cap\mathscr{S}(\R^2)$. 
 \item \label{item:FDKiv} Let us fix $M>0$ and consider $B_M\subset BC_H(\R^2)$, the closed ball of radius $M$ around the origin. Then for every $n\geq 0$ and $\alpha\in \Z_{\geq0}^2$ there exists a constant $C<\infty$ such that for every $F,G\in B_M$ and $|\lambda| \leq 1/2$ we have 
 \begin{equation}\label{hdc12}
    \sup_{\xi\in \R^2} \, \langle \xi\rangle^n\, \big |\partial^\alpha \big (p_F(\xi)-p_G(\xi)\big) \big |\, \leq C \, |\lambda|\, \Vert F-G\Vert_{L^\infty(\R^2)}\, . 
 \end{equation}
\end{enumerate}
\end{proposition}

\begin{proof} 
The proof will make use of many results and arguments which are presented in Section~\ref{sec:pseudodiffs}.

\begin{enumerate}[label={(\roman*)}, ref={(\roman*)}]
 \item This follows from the properties of $H_{b,\lambda,F}$, Corollary~\ref{cor:WF} and of the boundedness of $f$.
 \item Construct a smooth real-valued function $0\leq h\leq 1$ such that 
 \[ h\equiv1 \text{ on }  [\inf\sigma(H_{b,\lambda,F})-1, \infty) \quad \text{and} \quad  h\equiv0 \text{ on }(-\infty,\inf\sigma(H_{b,\lambda,F})-2]\,. \]
 Fix an arbitrary $k\geq 1$ and define $f_k(x)=h(x)\, (x-\iu)^k\, f(x)$. Then $f_k$ is a Schwartz function and using the support properties of $h$, we get
 \[ f(H_{b,\lambda,F}) = (H_{b,\lambda,F}-\iu)^{-k}\, f_k(H_{b,\lambda,F})\,. \]
 In view of Lemma~\ref{pseudo:resolvent2} we know that $(H_{b,\lambda,F}-\iu)^{-1}$ has a magnetic symbol in $S_0^{-2}(\R^2)$, hence due to Lemma~\ref{pseudo:moyal} we have that $(H_{b,\lambda,F}-\iu)^{-k}$ has a magnetic symbol in $S_0^{-2k}(\R^2)$. The same would be true for $f(H_{b,\lambda,F})$ if we could prove that $f_k(H_{b,\lambda,F})$ has a symbol in $S_0^{0}(\R^2)$. For that we exploit the ``matrix elements decay" criterion in \eqref{pseudo4}, with $R$ replaced by $f_k(H_{b,\lambda,F})$. We have 
 \begin{equation}\label{hdc15}
 \begin{aligned}
     \langle \Psi_{\alpha,\alpha^*}\, &,\, f_k(H_{b,\lambda,F})\, \Psi_{\beta,\beta^*}\rangle\\
     &=-\frac{1}{\pi}\int_{\R\times [-2,2]}\overline{\partial_z}(\tilde{f_k})_N(z)\langle \Psi_{\alpha,\alpha^*}\, ,(H_{b,\lambda,F}-z)^{-1}\, \Psi_{\beta,\beta^*}\rangle\, \di x\, \di y\,.
 \end{aligned}
 \end{equation}
For $j\in\{1,2$\}, recall that we denote by $X_j$ the multiplication operator in $L^2(\R^2)$ by $x_j$, and by  $\Pi_j\coloneq(\iu/2)\,[H_b,X_j]=-\iu\partial_{x_j}-b\,A_j(x)$ the $j$-th component of the magnetic momentum operator.
 Let $T_j$ be one of the operators $X_1$, $X_2$, $\Pi_1$ or $\Pi_2$. They leave $\mathscr{S}(\R^2)$ invariant. From Lemma~\ref{pseudo:resolvent2} we know that $(H_{b,\lambda,F}-z)^{-1}$ has a symbol $r_{z,F}\in S^{-2}_0(\R^2)$ and obeys \eqref{hdc4}. By using partial integration we may compute the following commutators on $\mathscr{S}(\R^2)$: 
 \begin{align*}
 [X_1,\mathrm{Op}^A(r_{z,F})] &= \iu \, \mathrm{Op}^A(\partial_{\xi_1}r_{z,F})\,, & [X_2,\mathrm{Op}^A(r_{z,F})] &= \iu \, \mathrm{Op}^A(\partial_{\xi_2}r_{z,F})\,, \\
 [\Pi_1, \mathrm{Op}^A(r_{z,F})] & = \iu \, b \, \mathrm{Op}^A(\partial_{\xi_2}r_{z,F})\, , & [\Pi_2, \mathrm{ Op}^A(r_{z,F})] &= -\iu \, b \, \mathrm{Op}^A(\partial_{\xi_1}r_{z,F})\,,
 \end{align*}
 which can be extended to bounded operators on $L^2(\R^2)$ by Lemma~\ref{pseudo:cald-vaill}. By a recursive procedure we obtain similar formulas for any possible multiple commutators of $T_j$ with the resolvent, so that they can be extended to bounded operators. Using the same strategy as in Corollary~\ref{corollary-beals}, we may transform multiple commutators with $T_j$ into decay in $\alpha-\beta$ and $\alpha^*-\beta^*$, the price being some growing power-like factors in $\langle z\rangle$, which are controlled by the decay of $f_k$, and growing factors in powers of $1/|{\rm Im}(z)|$, which are controlled by choosing $N$ large enough in the almost analytic extension \eqref{almost_analytic}. 
 \item It is a direct consequence of \ref{item:FDKii}. 
 \item In view of Corollary~\ref{cor:WF}, for $F\in B_M$ the operator norm of $W_F$ is bounded by $2 \, \Vert v\Vert_{L^1(\R^2)}\, M$. Thus, provided that $|\lambda| \le 1/2$, the spectrum of $H_{b,\lambda,F}$ is contained in the interval $[b-\Vert v\Vert_{L^1(\R^2)}\, M,\infty)$, hence the cut-off function $h$ from point \ref{item:FDKii} can be chosen independently of $F\in B_M$ and of $\lambda \in [-1/2,1/2]$.
 
 The operator $W_F-W_G$ is a magnetic pseudo-differential operator with a symbol 
 \[ w_{F-G}(\xi) = \left(\int v(y) \, \di y \right) \, \big (F(0)-G(0)\big )-(2\pi)^{-2}\int_{\R^2} \eu^{\iu \, \xi\cdot x}\, v(x) \, \big (F(x)-G(x)\big )\, \di x \]
 belonging to $S_0^0(\R^2)$.
 Due to the fast decay of $v$, all the seminorms of $w_{F-G}$ are bounded by $\Vert F-G\Vert_{L^\infty(\R^2)}$. In particular, for $G=0$ this implies that all the constants appearing in the seminorms of $w_F$ are uniformly bounded when $F\in B_M$, and the same holds for the constants appearing in the seminorms of the resolvent symbols~$r_{z,F}$.
 
 Now we consider the identity 
 \begin{align*} 
 f(H_{b,\lambda,F})-f(H_{b,\lambda,G})& =\Big ((H_{b,\lambda,F}-\iu)^{-k}-(H_{b,\lambda,G}-\iu)^{-k} \Big )\, f_k(H_{b,\lambda,F})\\
 &\quad + (H_{b,\lambda,G}-\iu)^{-k}\, \Big (f_k(H_{b,\lambda,F})-f_k(H_{b,\lambda,G})\Big ).
 \end{align*}
 We have 
 \begin{align*}
 (H_{b,\lambda,F}-\iu)^{-k} &-(H_{b,\lambda,G}-\iu)^{-k} = \Big ((H_{b,\lambda,F}-\iu)^{-1} -(H_{b,\lambda,G}-\iu)^{-1}\Big )(H_{b,\lambda,F}-\iu)^{-k+1}\\
 &\quad +(H_{b,\lambda,G}-\iu)^{-1}\Big ((H_{b,\lambda,F}-\iu)^{-1} -(H_{b,\lambda,G}-\iu)^{-1}\Big )(H_{b,\lambda,F}-\iu)^{-k+2}\\
 &\quad + \cdots + (H_{b,\lambda,G}-\iu)^{-k+1}\Big ((H_{b,\lambda,F}-\iu)^{-1} -(H_{b,\lambda,G}-\iu)^{-1}\Big )
 \end{align*}
 where each term, using the second resolvent identity and the composition rule for symbols from Lemma~\ref{pseudo:moyal}, will have a symbol in $S_0^{-2k-2}(\R^2)$ with all its seminorms bounded by $|\lambda|\,  \Vert F-G\Vert_{L^\infty(\R^2)}$.  
    
 It remains to prove that $f_k(H_{b,\lambda,F})-f_k(H_{b,\lambda,G})$ has a symbol in $S_0^0(\R^2)$ with all its seminorms bounded by $|\lambda|\, \Vert F-G\Vert_{L^\infty(\R^2)}$. Reasoning like in \eqref{hdc15} we need to show strong decay in $\alpha-\beta$ and $\alpha^*-\beta^*$ for scalar products of the type 
 \[\langle \Psi_{\alpha,\alpha^*}, (H_{b,\lambda,F}-z)^{-1} \lambda W_{F-G} (H_{b,\lambda,G}-z)^{-1}\, \Psi_{\beta,\beta^*}\rangle,\]
  by also controlling the behavior in $z$. The fact that we have a term $\lambda W_{F-G}$ in the middle provides us with a bound of the form $|\lambda|\,  \Vert F-G\Vert_{L^\infty(\R^2)}$ for all these scalar products, while by the same commutator method we obtain a decay in $\alpha-\beta$ and $\alpha^*-\beta^*$ which leads to a polynomial growth in $\langle z\rangle $ and in $|{\rm Im}(z)|^{-1}$, both of these being taken care of by the almost analytic extension of $f_k$. 
  
  We have thus proved that the difference of symbols in \eqref{hdc12} is smoothing, with all its seminorms bounded by $|\lambda| \, \|F-G\|_{L^\infty(\mathbb{R}^2)}$, and we are done.
  \qedhere 
\end{enumerate}
\end{proof}

\subsection{Existence of a fixed point}

In this subsection we prove a corollary which will be useful for the proof Theorem~\ref{thm1} and we use the notation introduced there. We are interested in the existence of a fixed point for the equation
\[
    \Phi_{\lambda}(F) = F, 
\]
with $F\in BC_H(\mathbb{R}^2)$, where 
\[ \Phi_{b,\lambda}^{f}(F)(\cdot ) \equiv \Phi_\lambda(F)(\cdot )=f(H_{b,\lambda,F})(\cdot,0)\,.\]

\begin{corollary}\label{thm:fix_point_final}
    For any $M>\Vert f(H_b)(\cdot,0)\Vert_{L^\infty(\R^2)}$, there exists $0<\lambda_0\leq 1$ such that the map $\Phi_\lambda$ has a unique fixed point in the closed ball $B_M$ of $BC_H(\R^2)$ for every $0\leq |\lambda|\leq \lambda_0$. The fixed points are Schwartz functions.
\end{corollary}

\begin{proof}
We have $f(H_b)(\cdot,0)=\Phi_\lambda(0)$ for all $\lambda\in \R$. With the notation in \eqref{hdc12}, by choosing $G=0$, $\alpha=(0,0)$ and $n=3$, Proposition~\ref{prop:fermi_dirac_kernel} \ref{item:FDKiv} implies that there exists $C<\infty$ such that for every $F\in B_M$ and $| \lambda|\leq 1/2$ we have 
\[|p_F(\xi)-p_0(\xi)|\leq C\, M\, |\lambda| \, \langle \xi\rangle^{-3}.  \]
By taking the inverse Fourier transform (see \eqref{eq: fhblambdaF}) and using the triangle inequality we obtain that if $|\lambda|$ is small 
enough then 
\[\Vert \Phi_\lambda(F)\Vert_{L^\infty(\R^2)}\leq \Vert \Phi_\lambda(0)\Vert_{L^\infty(\R^2)} + C'\, |\lambda|\, \leq M,\]
which shows that the closed ball $B_M$ is left invariant by $\Phi_\lambda$. 

The last ingredient we need in order to apply Banach's fixed point theorem is the contraction property, but this also follows from \eqref{hdc12} when $|\lambda|$ is sufficiently small. 
\end{proof}

\subsection{Proof of Theorem~\ref{thm1}}
In this subsection we complete the proof of Theorem~\ref{thm1}.

\begin{proof}[Proof of Theorem~\ref{thm1}]
The fact that, for any $F\in BC_H(\R^2)$, the operator $W_F$ is self-adjoint and bounded implies, in particular, that $H_{b,\lambda,F}$ is a lower-bounded self-adjoint operator with the same domain as $H_b$. This follows from Corollary~\ref{cor:WF}.

If $f$ satisfies \eqref{cdh2}, then the operator $f(H_{b,\lambda,F})$ admits a jointly continuous integral kernel $f(H_{b,\lambda,F})(x,x')$. Moreover, for every $F\in BC_H(\R^2)$ we have
\[
\Phi_\lambda(F)(\cdot)=f(H_{b,\lambda,F})(\cdot,0)\in BC_H(\R^2).
\]
These properties follow from Proposition~\ref{prop:fermi_dirac_kernel}, parts (ii) and (iii). In particular, for all $\lambda\in\R$,
\[
\Phi_\lambda\bigl(BC_H(\R^2)\bigr)\subseteq \mathscr{S}(\R^2).
\]

Finally, combining Proposition~\ref{prop:fermi_dirac_kernel}, part (iv), with Corollary~\ref{thm:fix_point_final}, we conclude that there exists $\lambda_0>0$, depending on $b$ and $f$, such that the fixed-point equation
\[
\Phi_\lambda(F)=F, \qquad F\in BC_H(\R^2),
\]
admits a solution for all $|\lambda|\le \lambda_0$.
\end{proof}

\section{Quantization of the self-consistent Hall conductivity at ``zero temperature"}\label{sec:zerotemp}

In order to prove Theorem~\ref{thm2}, we first prove the Proposition below, where for the moment we assume that $P_{b,\lambda,N}$ as defined in \ref{thm2-i} of Theorem~\ref{thm2} is an orthogonal projection. This fact will be proved later on.

\begin{proposition}\label{prop:traces}
    Let $|\lambda|\leq \lambda_1$, let $P_{b,\lambda,N}$ be defined as in \ref{thm2-i} of Theorem~\ref{thm2}, and let $\chi_L$ be the characteristic function of the disk $\Lambda_L\coloneq\{x\in\R^2\mid |x|<L\}$. Assume that $P_{b,\lambda,N}$ is an orthogonal projection. Then there exists $0<\lambda_2\leq \lambda_1$ such that for any $|\lambda|\leq \lambda_2$ we have that $\chi_LP_{b,\lambda,N}$ is trace class and 
    \[\lim_{L\rightarrow\infty}\frac{1}{|\Lambda_L|}|\tr(\chi_LP_{b,\lambda,N})-\tr(\chi_LP_{b,0,N})|=0.\]
\end{proposition}

Proposition~\ref{prop:traces} and its proof are heavily inspired by Lemma~2.1 and Appendix~C of \cite{CorneanMonacoMoscolari2021} and Proposition~3.3 and Appendix~A of \cite{CorneanMoscolari2025}.

\begin{proof}[Proof of Proposition~\ref{prop:traces}]
    From Proposition~\ref{prop:fermi_dirac_kernel} \ref{item:FDKiii} we conclude that the kernel $P_{b,\lambda,N}(x,y)=f_N(H_{b,\lambda,F_{b,\lambda}})(x,y)$ is localized around the diagonal in the sense that for every $n\geq1$ there exists $C>0$ such that
    \begin{equation*}
      |P_{b,\lambda,N}(x,y)|\leq C \langle x-y\rangle^{-n},\quad \forall x,y\in \R^2.
    \end{equation*}
    This implies that $\chi_LP_{b,\lambda,N}\langle\cdot\rangle^3$ and $\langle\cdot\rangle^{-3} P_{b,\lambda,N}$ are Hilbert-Schmidt operators, hence their product $\chi_LP_{b,\lambda,N}$ is trace class.

    Furthermore, Proposition~\ref{prop:fermi_dirac_kernel} \ref{item:FDKiv} leads (via an inverse Fourier transform) to the fact that for every $n\geq1$ there exists $C>0$ such that
    \begin{equation*}
      |P_{b,\lambda,N}(x,y)-P_{b,0,N}(x,y)|\leq C\, |\lambda| \, \langle x-y\rangle^{-n},\quad \forall x,y\in \R^2.
    \end{equation*}
    By the Schur test, the above estimate gives that if $|\lambda|$ is smaller than some $\lambda_2\leq \lambda_1$  then $\Vert P_{b,\lambda,N}-P_{b,0,N}\Vert \leq 1/2$ and we may construct the Kato-Nagy intertwining unitary \cite{Kato1966}
    \[U_\lambda=\big ( {\rm Id} -(P_{b,\lambda,N}-P_{b,0,N})^2\big )^{-1/2}\big(P_{b,\lambda,N}P_{b,0,N}+({\rm Id}-P_{b,\lambda,N})({\rm Id}-P_{b,0,N})\big),\]
    so that $P_{b,\lambda,N}U_\lambda=U_\lambda P_{b,0,N}$. Moreover, as we show next, $V_\lambda\coloneq U_\lambda-{\rm Id}$ is a magnetic pseudo-differential operator with smoothing symbol (term introduced in Proposition~\ref{prop:fermi_dirac_kernel} \ref{item:FDKii}), and so it has an integral kernel with the same decay properties as $P_{b,\lambda,N}$. 
    
    We first write:
    \begin{align*}
        V_\lambda=\Big(\big ( {\rm Id} &-(P_{b,\lambda,N}-P_{b,0,N})^2\big )^{-1/2}-{\rm Id}\Big)\big(P_{b,\lambda,N}P_{b,0,N}+({\rm Id}-P_{b,\lambda,N})({\rm Id}-P_{b,0,N})\big)\\
        &+\big(P_{b,\lambda,N}P_{b,0,N}+({\rm Id}-P_{b,\lambda,N})({\rm Id}-P_{b,0,N})\big)-{\rm Id}
    \end{align*}
    and we will concentrate on the factor \[\big ( {\rm Id} -(P_{b,\lambda,N}-P_{b,0,N})^2\big )^{-1/2}-{\rm Id}.\] Indeed, if we can show that this factor has a smoothing symbol, it then follows from Lemma~\ref{pseudo:moyal} that $V_\lambda$ also has one.

    Define $h(z)\coloneq(1-z^2)^{-1/2}-1$ on the unit open ball in $\Co$. This is a holomorphic function where $h(0)=h'(0)=0$, so for $w\in\Co$ with $|w|\leq1/2$ we have
    \[h(w)=(2\pi i)^{-1}\int_{|z|=\frac{2}{3}}\frac{h(z)}{z-w}\,\di z=(2\pi i)^{-1}\int_{|z|=\frac{2}{3}}\frac{h(z)w}{z(z-w)}\,\di z\]
    since $1/(z-w)=1/z+w/(z(z-w))$ and $h(z)/z$ has a removable singularity at $0$. Thus
    \begin{equation}\label{mht1}
        \begin{aligned}
            \big ( {\rm Id} &-(P_{b,\lambda,N}-P_{b,0,N})^2\big )^{-1/2}-{\rm Id}=h(P_{b,\lambda,N}-P_{b,0,N})\\
            &=(2\pi i)^{-1}(P_{b,\lambda,N}-P_{b,0,N})\int_{|z|=\frac{2}{3}}\frac{h(z)}{z}\big (z-(P_{b,\lambda,N}-P_{b,0,N})\big )^{-1}\,\di z.
        \end{aligned}
    \end{equation}
    By using the commutator criterion from Corollary~\ref{corollary-beals} as in the proof of Lemma~\ref{pseudo:resolvent} we can conclude that $\big (z-(P_{b,\lambda,N}-P_{b,0,N})\big )^{-1}$ is a magnetic pseudo-differential operator with symbol in $S_0^0(\R^2)$ when $|z|=\frac{2}{3}$. Furthermore, using the commutator method used in Lemma~\ref{pseudo:resolvent2} where instead of $H$ we use  $P_{b,\lambda,N}-P_{b,0,N}$, we get a uniform bound on the seminorms of these symbols, which implies that 
    \[\int_{|z|=\frac{2}{3}}\frac{h(z)}{z}\big (z-(P_{b,\lambda,N}-P_{b,0,N})\big )^{-1}\,\di z\]
    is a magnetic pseudo-differential operator with symbol in $S_0^0(\R^2)$. Then \eqref{mht1} together with Lemma~\ref{pseudo:moyal} shows that $\big ( {\rm Id} -(P_{b,\lambda,N}-P_{b,0,N})^2\big )^{-1/2}-{\rm Id}$, and consequently also $V_\lambda$, has a smoothing symbol.

    Now using the trace cyclicity we get:
    \begin{align*}
        \tr(\chi_LP_{b,\lambda,N})-\tr(\chi_LP_{b,0,N})&=\tr(\chi_LU_\lambda P_{b,0,N}U_\lambda^*)-\tr(\chi_LP_{b,0,N})=\tr([\chi_L,U_\lambda] P_{b,0,N}U_\lambda^*)\\
        &=\tr([\chi_L,V_\lambda] P_{b,0,N}V_\lambda^*)-\tr([\chi_L,V_\lambda] P_{b,0,N})\,.
    \end{align*}
    The traces of the operators $[\chi_L,V_\lambda] P_{b,0,N}V_\lambda^*$ and $[\chi_L,V_\lambda] P_{b,0,N}$ equal the integral of the diagonal values of their integral kernels. This diagonal value for $[\chi_L,V_\lambda] P_{b,0,N}$ is given by: 
    \[\Big ([\chi_L,V_\lambda] P_{b,0,N}\Big )(x,x)=\int_{\R^2}\di y\,(\chi_L(x)-\chi_L(y))V_\lambda(x,y)P_{b,0,N}(y,x)\]
    and something similar for the other operator where $P_{b,0,N}$ is replaced by $P_{b,0,N}V_\lambda^*$. By the previously discussed decay properties, both traces can then be bounded by a constant times the following integral: 
    \begin{align*}
        \int_{\R^4}\di x\,\di y\,|\chi_L(x)&-\chi_L(y)|\langle x-y\rangle^{-6}\\
        &=\int_{|x|< L,|y|\geq L}\di x\,\di y\,\langle x-y\rangle^{-6}+\int_{|x|\geq L,|y|< L}\di x\,\di y\,\langle x-y\rangle^{-6}\, .
    \end{align*}
    We have
    \[\int_{|x|< L,|y|\geq L}\di x\,\di y\,\langle x-y\rangle^{-6}\leq\int_{|x|< L}\di x\,\langle {\rm dist}(x,\partial\Lambda_L)\rangle^{-3}\int_{\R^2}\di y\,\langle y\rangle^{-3}\]
    and
    \[\int_{|x|< L}\di x\,\langle {\rm dist}(x,\partial\Lambda_L)\rangle^{-3}=C\int_0^L\di r \,r\,\langle L-r\rangle^{-3}\leq CL\int_0^\infty \di r \,\langle r\rangle^{-3}.\]
    Thus
    \[\tr(\chi_LP_{b,\lambda,N})-\tr(\chi_LP_{b,0,N})=\Og(L).\qedhere\]
\end{proof}

\subsection{Proof of Theorem~\ref{thm2}}

We now prove Theorem~\ref{thm2}.

\begin{proof}
We start by proving \ref{thm2-i}. For an $f_N$ which obeys all the conditions in the statement of the Theorem we have that $f_N(H_b)\equiv P_{b,0,N}$ is an orthogonal projection which corresponds to $N$ Landau levels. 
This projection is a magnetic pseudo-differential operator with a smooth integral kernel $P_{b,0,N}(x,x')$ by Proposition~\ref{prop:fermi_dirac_kernel}.

Now let us fix some $M>\Vert P_{b,0,N}(\cdot,0)\Vert_{L^\infty(\R^2)}$. Let $B_M\subset BC_H(\R^2)$ be the closed ball of radius $M$ around the origin. For $F\in B_M$ we consider (see \eqref{1sec})  
\[
W_F= \left(\int v(y) \, \di y\right) \,F(0) +Z_F,
\]
with
\[
Z_{F}(x,x') = - \eu^{\iu \, b \, \phi(x,x')} \, v(x-x')  F(x-x^\prime).
\]
Then, the operator norm of $\lambda W_F$ is uniformly bounded (up to a numerical constant) by $|\lambda|\, M$. By regular perturbation theory we obtain that, if $|\lambda|\leq \lambda_1$ with $0<\lambda_1\leq \lambda_0$ and with $\lambda_0$ from Theorem~\ref{thm1}, the set on which $f'_N\neq 0$ will still belong to the resolvent set of $H_{b,\lambda,F}=H_b+\lambda W_F$, hence $P_{b,\lambda,N,F}\coloneq f_N(H_{b,\lambda,F})$ will continue to be an orthogonal projection, but this time giving a spectral projection for $H_{b,\lambda,F}$. If $|\lambda|\leq \lambda_1$, we have a fixed point $F_{b,\lambda}$ of $\Phi_\lambda$ with $\Phi_\lambda(F) = f_N(H_{b,\lambda,F})$. Then 
\[ P_{b,\lambda,N}=f_N(H_{b,\lambda,F_{b,\lambda}})
\]
is the self-consistent ``Fermi projection".

We now prove \ref{thm2-ii}. Since $P_{b,\lambda,N}$ commutes with all magnetic translations, its integrated density of states is 
\[\mathcal{I}(b,\lambda)=P_{b,\lambda,N}(0,0).\]

We will show that $\mathcal{I}(b,\lambda)=\mathcal{I}(b,0)=Nb/(2\pi)$ for all $|\lambda|$ small enough. By Proposition~\ref{prop:traces} we can deduce that $\mathcal{I}(b,\lambda)=\mathcal{I}(b,0)$ for $|\lambda|\leq \lambda_2$  from
\[ \mathcal{I}(b,\lambda)=\lim_{L\rightarrow\infty}\frac{1}{|\Lambda_L|}\big(\tr(\chi_LP_{b,\lambda,N})-\tr(\chi_LP_{b,0,N})+\tr(\chi_LP_{b,0,N})\big)=\mathcal{I}(b,0). \qedhere \]
\end{proof}

\bigskip 
\begin{small}
\noindent{\bf Acknowledgements.} H.\ C.\ acknowledges support from Grant DFF–5281-00046B of the Independent Research Fund Denmark $|$ Natural Sciences. H.\ C.\ and M.\ H.\ T.\ also acknowledge support from Danish National Research Foundation (DNRF), through the Center CLASSIQUE, grant nr. 187. 

E.\ L.\ G.\ was partially funded by the Deutsche Forschungsgemeinschaft (DFG, German Research Foundation) through TRR 352 – Project-ID 470903074. 

D.\ M.\ gratefully acknowledges financial support from Ministero dell’Università e della Ricerca (MUR, Italian Ministry of University and Research) and Next Generation EU within PRIN 2022AKRC5P ``Interacting Quantum Systems: Topological Phenomena and Effective Theories'' and within PNRR--MUR Project no.~PE0000023-NQSTI. The work of D.\ M.\ was also supported by Sapienza Università di Roma within Progetto di Ricerca di Ateneo 2023 and 2024.

The work of E.\ L.\ G.\ and D.\ M.\ has been carried out under the auspices of the GNFM-INdAM (Gruppo Nazionale per la Fisica Matematica — Istituto Nazionale di Alta Matematica).

\end{small}
\end{document}